\documentclass[aps,superscriptaddress,nofootinbib,a4paper,longbibliography,twocolumn]{revtex4-2}
\usepackage{amsmath}
\usepackage{amsfonts}
\usepackage{amsthm}
\usepackage{amssymb}
\usepackage{graphicx}
\usepackage{enumerate}
\usepackage{color}
\usepackage[T1]{fontenc}
\usepackage{braket}
\usepackage{todonotes}
\usepackage{soul}
\usepackage{subfigure}
\usepackage{mathrsfs}
\usepackage{float}
\usepackage{amsmath}
\usepackage{bm}
\usepackage{multirow}
\usepackage{tabularx}
\usepackage{outlines}

\graphicspath{{figures/}}
\newtheorem{theorem}{Theorem}
\newtheorem{lemma}{Lemma}

\newtheorem*{theorem*}{Theorem}
\newtheorem{definition}{Definition}
\newtheorem{example}{Example}

\newtheorem{assumption}{Assumption}

\newtheoremstyle{operation}
  {\topsep}
  {\topsep}
  {}
  {}
  {\itshape}
  {.}
  {.5em}
  {\thmname{#1}\thmnumber{ #2}\thmnote{ (#3)}}
\theoremstyle{operation}

\def\>{\rangle}
\def\<{\langle}

\newcounter{protocol}
\renewcommand{\theprotocol}{\arabic{protocol}}

\usepackage[bookmarks=false,colorlinks,citecolor=blue,
linkcolor=blue,anchorcolor=blue,urlcolor=blue
]{hyperref}

\begin{document}
\title{Jamming Extensions of Quantum Correlations Lead to Hidden Superluminal Signaling}
\author{Ravishankar Ramanathan}
\email{ravi@cs.hku.hk}
\affiliation{School of Computing and Data Science, The University of Hong Kong, Pokfulam Road, Hong Kong}
\author{Xie Sicheng}
\affiliation{School of Computing and Data Science, The University of Hong Kong, Pokfulam Road, Hong Kong}
\author{Micha{\l} Eckstein}
\affiliation{Institute of Theoretical Physics,  Faculty of Physics, Astronomy and Applied Computer Science, 
Jagiellonian University,  ul. {\L}ojasiewicza 11,  30-348 Krak\'{o}w,
Poland}
\author{Pawe{\l} Horodecki}
\affiliation{International Centre for Theory of Quantum Technologies (ICTQT), University of Gda\'{n}sk,
Jana Ba\.{z}y\'{n}skiego 8, 80-309 Gda\'{n}sk}
\affiliation{Faculty of Applied Physics and Mathematics, Gda\'{n}sk University of Technology, Gabriela Narutowicza 11/12,
80-233 Gda\'{n}sk, Poland}

\begin{abstract}
    Relativistic causality in certain spacetime configurations permits jamming - superluminal causal influences that nevertheless do not enable superluminal signaling. We fully characterize the correlations compatible with relativistic causality (RC) by extending the operator framework of PRL 104, 140404. When the underlying state is required to be quantum, we prove that any nontrivial jamming - whether state-independent or state-dependent - necessarily leads to hidden superluminal signaling. Thus, the only consistent possibilities are standard quantum correlations without jamming or the full relativistically causal correlation set, with intermediate jamming extensions of quantum correlations leading to signaling. As an application of our characterization, we investigate the security of device-independent (DI) cryptographic primitives against adversaries constrained only by relativistic causality. We construct explicit RC attacks that break DI bit commitment and DI secret sharing in jamming geometries, even when these protocols remain secure against no-signaling adversaries. Our results show that security against relativistic adversaries requires behaviors to be specified together with the spacetime locations of their measurement events; input-output statistics alone are insufficient.  
\end{abstract}

\maketitle

\textit{Introduction.-} Relativistic causality and the no-signaling principle are often treated as interchangeable constraints on physical correlations. However, it has been recently discovered that in multipartite Bell scenarios, they are not equivalent \cite{GPR96, HR19, VC22}. Depending on the spacetime locations of the measurement events, relativistic causality only enforces a subset of the usual no-signaling equalities. In particular in certain configurations, this permits jamming - a form of superluminal causal influence that alters the joint correlations of a subset of parties while leaving individual marginals untouched, thereby preserving the relativistic causal constraint of no superluminal signaling \cite{GPR96, HR19, VC22}. 

This leads to a fundamental question: can quantum correlations exploit the causal loophole left open under relativistic causality? In other words, do there exist extensions of quantum correlations (intermediate between the quantum set and the full relativistic causal set) that allow nontrivial jamming while remaining compatible with quantum kinematics and the impossibility of superluminal signaling. In this paper, we answer this question in the negative. We first extend the unified operator framework for multipartite correlations (including no-signaling, quantum and classical correlations) by Acin et al. in \cite{Acin10} and characterize the maximal set of correlations compatible with relativistic causality.  This $\mathcal{RC}$ set is expressed in this formalism using a Hermitian operator $\Omega$ and locally identity-preserving  jamming maps $\Phi$, is operationally closed and free of any signaling pathologies. As such, it is a well-defined object (formally a convex polytope) consistent with relativistic causality and worthy of study in a similar fashion to the no-signaling polytope \cite{MAG06, PR94, PBS11}. We then show that when the underlying state is required to be a valid density operator (i.e., when $\Omega$ is required to be a positive semidefinite $\rho$), any attempt to introduce nontrivial jamming maps (in either a state-independent or state-dependent manner) necessarily generates hidden superluminal signaling. This argument is an analog of the classic result by Gisin \cite{Gisin90, SBG01, BH15} on the incompatibility of state-dependent maps with no-signaling. Specifically here by hidden superluminal signaling, we mean that a remote party can choose between two ensemble decompositions of the same density operator, and that this choice produces distinguishable output states after a state-dependent jamming evolution. The remote measurement choice  can therefore be inferred outside its future light cone, even though the original input-output statistics satisfy the relativistic causality marginal constraints, giving rise to superluminal signaling. The remaining consistent possibilities are therefore the standard set of quantum correlations without jamming or the full relativistically causal correlation set, with intermediate sets leading to signaling.

The set of relativistically causal correlations has cryptographic consequences. Device-Independent (DI) protocols aim to achieve security solely based on the input-output behavior of devices without placing on any trust on their inner workings or quantum characterization. The ultimate security goal here is to achieve security against adversaries only limited by the principle of no-superluminal-signaling. One primary motivation for studying the set $\mathcal{RC}$ has been its implications \cite{HR19, SKGH+21} for the security of DI randomness generation \cite{Colbeck07, Brandao16, RLW25} and key distribution \cite{BHK05}, where attacks have been found in $\mathcal{RC}$ theories exploiting a lack of monogamy of nonlocality. In this work, we investigate the implications for the relativistic security of a different class of primitives, namely mistrustful cryptographic protocols, in which one or more protocol participants may themselves actively deviate from the prescribed strategy. DI protocols have been established for primitives like bit commitment (by Fehr and Fillinger \cite{FF15}) and secret sharing (by Moreno et al. \cite{MBNC20}), and proven secure against no-signaling adversaries under the assumption of independent and identically distributed (i.i.d.) behavior over multiple rounds. We show that these protocols can be rendered insecure in measurement configurations that admit the possibility of jamming. Specifically, we construct explicit relativistic attacks on these protocols, demonstrating that \(\mathcal{RC}\) correlations allow active cheating in dishonest cryptography rather than just passive leakage to an external adversary. Notably, these results are not consequences merely of the aforementioned lack of monogamy of nonlocality. These attacks reinforce the notion that input-output statistics alone are insufficient for security against relativistic adversaries, the behaviors must also carry the spacetime labels of their constituent measurement events.

\textit{A Unified Framework for Relativistic Causal, No-Signaling, Quantum and Classical Correlations.-}

We begin by recalling the set of correlations of interest. For convenience and to avoid cluttering notation, our exposition throughout the paper will mainly focus on the three-player scenario with the generalization to the arbitrary $n$-player scenario being deferred to Appendix \ref{sec:RC}. Accordingly, we consider a tripartite Bell experiment with three players Alice, Bob and Charlie who, at spacelike separated points $(t_A, p_A)$, $(t_B, p_B)$ and $(t_C, p_C)$ freely and randomly choose measurements labeled by $x, y, z$ with $x \in [m_A], y \in [m_B], z \in [m_C]$ and obtain outcomes $a \in [r_A], b \in [r_B], c \in [r_C]$ respectively, where we adopt the notation $[n] := \{1,2,\ldots,n\}$. The experiment is characterized by the set of joint probability distributions $P(a,b,c|x,y,z)$ that together constitute the behavior/correlation $\{P(a,b,c|x,y,z)\}$. The set of no-signaling correlations \cite{MAG06, PR94, PBS11} are those that obey the conditions that the marginal distributions seen by any subset of players is independent of the inputs chosen by the complementary set of players. Accordingly, we denote the set $\mathcal{NS}$ as those behaviors $\{P(a,b,c|x,y,z)\}$ for which the Alice-Bob marginal $P(a,b|x,y)$ is independent of $z$ for all $a,b,x,y,z$, and likewise for all Alice-Charlie and Bob-Charlie marginal probabilities. This set of behaviors $\mathcal{NS}$ constitutes a convex polytope of dimension $dim(\mathcal{NS}) = \prod_{K=A,B,C} [1 + m_K(r_K-1)] - 1$. 

As shown in \cite{GPR96, HR19, VC22, EMHRH25}, the relativistic causality conditions that enforce no superluminal signaling are not precisely captured by the above no-signaling constraints in general multipartite scenarios. In particular, in certain spacetime configurations of the measurement events that are termed jamming configurations, only a subset of the no-signaling conditions need to be considered. In a jamming configuration such as shown in Fig. \ref{fig:jamming-spacetime}, where the intersection of the future light-cones of Alice's and Charlie's output events lies inside the future light-cone of Bob's input event, relativistic causality permits the joint correlations of Alice and Charlie to depend on Bob's input, provided their single-party marginals remain independent of that input. While detailed derivations and explanations can be found in \cite{GPR96, HR19, VC22, EMHRH25}, intuitively the reason is clear. The joint correlations of Alice and Charlie's outcomes are only accessible at a spacetime point in the intersection of their future light-cones. As such, any superluminal causal influence that jams their correlations still does not lead to superluminal signaling - no faster-than-light transmission of information takes place, and the resulting set of correlations remains compatible with relativistic causality. This larger set of behaviors $\{P(a,b,c|x,y,z)\}$ compatible with this weaker set of conditions is termed the set of relativistically causal correlation $\mathcal{RC}$. The set $\mathcal{RC}$ is again a convex polytope, this time of dimension $dim(\mathcal{RC}) = dim(\mathcal{NS}) + (m_B-1)m_A(r_A-1)m_C(r_C-1)$ \cite{HR19, SKGH+21}.

\begin{figure}[t]
    \centering
    \includegraphics[width=1.1\columnwidth]{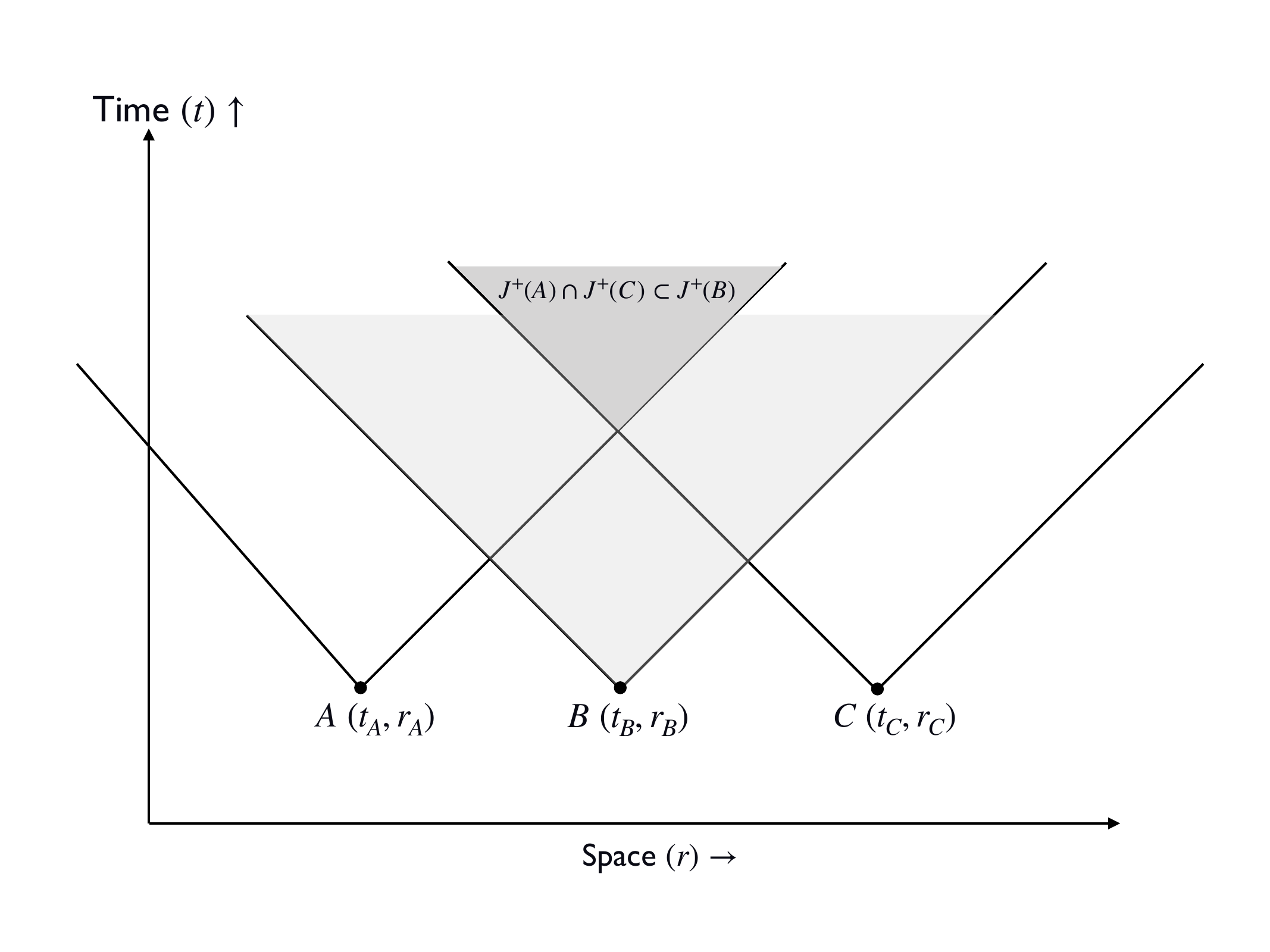} 
    \caption{The Jamming spacetime configuration for three spacelike separated players at locations $A (t_A, r_A)$, $B ( t_B, r_B)$ and $C (t_C, r_C)$. In this configuration, the intersection of the future light cones of $A$ and $C$, denoted $J^+(A) \cap J^+(C)$ is entirely contained within the future light cone of $B$ denoted $J^+(B)$. }
    \label{fig:jamming-spacetime}
\end{figure}

In \cite{Acin10}, Ac\'{i}n et al. formulated a unified framework of no-signaling, quantum and classical correlations in terms of an elegant operator representation. Our first result is an extension of this framework to include the set of relativistically causal correlations. In particular, we characterize the set of $\mathcal{RC}$ correlations Thm. \ref{thm:RC-rep}, with a detailed proof provided in Appendix \ref{sec:RC}.  
\begin{theorem}
\label{thm:RC-rep}
An arbitrary tripartite behavior $\{P(a,b,c\vert{}x,y,z)\}$ satisfies the Relativistic Causality $\mathcal{RC}$ constraints for a jamming space-time configuration (where the intersection of Alice and Charlie's future light cones is contained within Bob's future light cone) if and only if there exist:  
\begin{itemize}
    \item Local Hilbert spaces $\mathcal{H}_A, \mathcal{H}_B, \mathcal{H}_C$ and sets of local quantum measurements $\{M_a^x\}, \{M_b^y\}, \{M_c^z\}$ (that sum to their respective local identities, e.g., $\sum_a M_a^x = \mathbb{I}_A$), 
    \item  a Hermitian operator $\Omega_{ABC} \in \mathcal{L}(\mathcal{H}_A \otimes \mathcal{H}_B \otimes \mathcal{H}_C)$ of unit trace ($\text{tr}(\Omega_{ABC}) = 1)$, and
    \item a family of linear, trace-preserving maps $\Phi_y : \mathcal{L}(\mathcal{H}_A \otimes \mathcal{H}_C) \to \mathcal{L}(\mathcal{H}_A \otimes \mathcal{H}_C)$, parameterized by Bob's setting $y$,
\end{itemize}  
such that:
\begin{equation}
    \label{eq:gen-trace-rule}
    P(a,b,c\vert{}x,y,z) = \text{tr}\Big[ \Omega_{ABC} \big( \Phi_y(M_a^x \otimes M_c^z) \otimes M_b^y \big) \Big],
\end{equation}
for all $a,b,c,x,y,z$.
And the map $\Phi_y$ for any $y$ is locally identity-preserving, i.e., its action on any local observable $X_A \in \mathcal{L}(\mathcal{H}_A)$ and $X_C \in \mathcal{L}(\mathcal{H}_C)$ satisfies:
\begin{eqnarray}
\label{eq:Phiy-RC}
\Phi_y(X_A \otimes \mathbb{I}_C) &=& X_A \otimes \mathbb{I}_C, \nonumber \\ \Phi_y(\mathbb{I}_A \otimes X_C) &=& \mathbb{I}_A \otimes X_C.
\end{eqnarray}
\end{theorem}
Here $\mathcal{L}(\mathcal{H}_A \otimes H)$ denotes the set of all linear operators from the composite Alice-Charlie Hilbert space to itself.
When all maps $\Phi_y$ are the identity ($\Phi_y = \mathbb{I}_{AC})$, this expression recovers the operator representation of no-signaling correlations from \cite{Acin10}. To emphasize the power of the operator framework and to show the distinction between $\mathcal{RC}$ and $\mathcal{NS}$, we derive in Appendix \ref{sec:RC} the operator representation of an example correlation that belongs to $\mathcal{RC}$ but is outside the set $\mathcal{NS}$. We end this section by emphasizing that the set $\mathcal{RC}$ is convex, free of signaling pathologies, and operationally constitutes the natural relativistically causal generalization of the no-signaling polytope to arbitrary spacetime configurations of measurement events, and that Thm. \ref{thm:RC-rep}  characterizes this set in an elegant and unified operator framework. 

\textit{Jamming extensions of quantum correlations.-}
The physical significance of jamming has been subject to debate recently. In \cite{VC22} it was argued on the basis of certain fine-tuned configurations that even the subset of relativistically causal constraints that we have outlined above might not rule out causal loops in $1+1$ Minkowski spacetime, while in \cite{Weilenmann25} it was argued on the basis of certain entropic monogamy relations obeyed in such theories that they might lead to superluminal signaling, see \cite{EMHRH25} as well for clarifications and counter arguments. In this section, we consider the more natural question - if relativistic causality permits jamming, is it possible to extend quantum correlations to incorporate jamming, and are such extensions physical and operationally meaningful? In this context, we remark that besides the no-signaling set which has been studied for its physicality in terms of communication complexity \cite{Brassard06}, information causality \cite{Pawlowski09} and other reasonable information-theoretic principles \cite{Fritz13}, so have intermediate sets such as the `almost quantum' set \cite{NGHA15}. As such, it is of significant interest to identify whether a similar intermediate set, extending quantum correlations to incorporate jamming while respecting relativistic causality, is physically plausible and to identify information-theoretic principles to rule out jamming. 

Accordingly, we consider the situation wherein the physical state of the system $\Omega_{ABC}$ in Thm. \ref{thm:RC-rep} is required to be a positive semidefinite Hermitian operator of unit trace $\rho_{ABC} \in \mathcal{L}(\mathcal{H}_A \otimes \mathcal{H}_B \otimes \mathcal{H}_C)$ ($\rho_{ABC} \ge 0$ and $\text{tr}(\rho_{ABC}) = 1$). That is, we consider the set of correlations obtained when in our framework we restrict ourselves to the set of states from standard quantum theory while still allowing for the jamming maps. For technical reasons which we will elaborate on later, we also restrict the local Hilbert spaces $\mathcal{H}_A, \mathcal{H}_B, \mathcal{H}_C$ to be finite-dimensional. We then obtain an extension of quantum theory where the joint probability of outcomes given settings is given by:
$P(a,b,c\vert{}x,y,z) = \text{tr}[\rho_{ABC}(\Phi_{y}(M_{a}^{x}\otimes M_{c}^{z})\otimes M_{b}^{y})]$,
which is identical to quantum theory except for the addition of the jamming map $\Phi_y$. Now to ensure non-negativity of the probabilities $p(a,b,c|x,y,z)$, we must impose that the maps obey $\left((\Phi_y)^{\dagger}_{AC} \otimes \mathbb{I}_B \right)  (\rho_{ABC}) \succeq 0$, i.e., $\Phi_y^{\dagger}(\rho_{AC}) \succeq 0$. 
We term the resulting correlations, Relativistically Causal Quantum ($\mathcal{RCQ}$) correlations. Formally, we have the following.

\begin{definition}
\label{def:RCQ}
    A tripartite behavior $P(a,b,c|x,y,z)$ belongs to the set of relativistically-causal quantum ($\mathcal{RCQ}$) correlations for a jamming space-time configuration if and only if there exist:
    \begin{itemize}
      \item local finite-dimensional Hilbert spaces $\mathcal{H}_A, \mathcal{H}_B, \mathcal{H}_C$ and sets of local quantum measurements $\{M_a^x\}, \{M_b^y\}, \{M_c^z\}$ (that sum to their respective local identities, e.g., $\sum_b M_b^y = \mathbb{I}_B$), 
    \item  a positive semidefinite Hermitian operator $\rho_{ABC} \geq 0$ in $\mathcal{L}(\mathcal{H}_A \otimes \mathcal{H}_B \otimes \mathcal{H}_C)$ of unit trace ($\text{tr}(\rho_{ABC}) = 1)$,
    \item a family of linear maps $\Phi_y : \mathcal{L}(\mathcal{H}_A \otimes \mathcal{H}_C) \to \mathcal{L}(\mathcal{H}_A \otimes \mathcal{H}_C)$, parameterized by Bob's setting $y$, that are trace-preserving $\Phi_y^{\dagger}(\mathbb{I}_A \otimes \mathbb{I}_C) = \mathbb{I}_A \otimes \mathbb{I}_C$, and preserve local identities, i.e., their action on any local observable $X_A \in \mathcal{L}(\mathcal{H}_A)$ and $X_C \in \mathcal{L}(\mathcal{H}_C)$ satisfies:
$$\Phi_y(X_A \otimes \mathbb{I}_C) = X_A \otimes \mathbb{I}_C,$$
$$\Phi_y(\mathbb{I}_A \otimes X_C) = \mathbb{I}_A \otimes X_C,$$ and furthermore, the maps $\Phi_y$ preserve positivity of the initial state, i.e., $\Phi_y^{\dagger}(\rho_{AC}) \succeq 0$,
\end{itemize}  
such that:
\begin{equation}
    \label{eq:gen-trace-rule}
    P(a,b,c\vert{}x,y,z) = \text{tr}\Big[ \rho_{ABC} \big( \Phi_y(M_a^x \otimes M_c^z) \otimes M_b^y \big) \Big].
\end{equation}

\end{definition}

The set of correlations $\mathcal{RCQ}$ in Def. \ref{def:RCQ} is a nontrivial extension of quantum correlations in that there exist RCQ correlations outside the standard quantum set. In Appendix \ref{sec:JammingQuantum-signal}, we present a paradigmatic example of a correlation in $\mathcal{RCQ}$ (together with the corresponding state $\rho_{ABC}$ and jamming maps $\Phi_y$) that is not in the quantum set. Specifically, in the example Bob jams the correlations between Alice-Charlie, switching from a maximally nonlocal correlation for $y=0$ (violating the CHSH inequality up to its quantum maximum value) to a fully mixed behavior for $y=1$.


It is important to note that the distinction with standard quantum theory arises here because of the action of jamming maps $\Phi_y$ which are linear trace-preserving but not necessarily CPTP or even just positive. The natural question arises on the physicality of such jamming extensions of quantum theory.
Our first result in this section that while the jamming maps (or more precisely their adjoint maps) are required to preserve positivity of the initial state $\rho_{AC}$, relaxing them to be positive on all states exactly recovers quantum theory. We formalize this via the following theorem.

\begin{theorem}
\label{thm:jammingmap-pos}
Let the linear trace-preserving maps $\Phi_y: \mathcal{L}(\mathcal{H}_A \otimes \mathcal{H}_C) \rightarrow \mathcal{L}(\mathcal{H}_A \otimes \mathcal{H}_C)$ satisfy the Relativistic Causality constraints in \eqref{eq:Phiy-RC}. If $\Phi_y$ are additionally required to be positive maps for all $y$, then $\Phi_y \; \forall y$ is strictly the identity map on the joint space $\mathcal{L}(\mathcal{H}_A \otimes \mathcal{H}_C)$, and the behaviors reduce to those obtained in quantum theory.
\end{theorem}  
\begin{proof}
First observe that since $\Phi_y$ is a Positive map, its adjoint $\Phi_y^\dagger$ is also a Positive map. This ensures that for any positive semidefinite operator $\rho_{AC}$, the mapped operator $\Phi_y^\dagger(\rho_{AC})$ remains positive semidefinite.

Now, the Relativistic Causality constraints require $\Phi_y$ to be locally identity-preserving as per \eqref{eq:Phiy-RC}.
This gives
\begin{eqnarray}
    \text{tr}[\Phi_y^\dagger(\rho_{AC}) (X_A \otimes \mathbb{I}_C)] &=& \text{tr}[\rho_{AC} \Phi_y(X_A \otimes \mathbb{I}_C)] \nonumber \\
    &=& \text{tr}[\rho_{AC} (X_A \otimes \mathbb{I}_C)],
\end{eqnarray}
showing that $\text{tr}_C[\Phi_y^\dagger(\rho_{AC})] = \text{tr}_C[\rho_{AC}]$. By similar reasoning for $X_C$, we also have $\text{tr}_A[\Phi_y^\dagger(\rho_{AC})] = \text{tr}_A[\rho_{AC}]$. Thus, the adjoint map $\Phi_y^\dagger$ perfectly preserves the marginal reduced density matrices of any bipartite input operator.

Now, consider the action of $\Phi_y^\dagger$ on an arbitrary pure product state $\Pi_{AC} = \vert{}\psi\rangle\langle\psi\vert{}_A \otimes \vert{}\phi\rangle\langle\phi\vert{}_C$. Since $\Phi_y^\dagger$ is a positive map, the resulting operator $\tilde{\rho}_{AC} = \Phi_y^\dagger(\Pi_{AC})$ is positive semidefinite. Applying the marginal preservation property, we obtain that $\text{tr}_C[\tilde{\rho}_{AC}]=\vert{}\psi\rangle\langle\psi\vert{}_A.$
Since the bipartite state has a perfectly pure marginal, the joint state must factor as a tensor product and the output must take the form $\tilde{\rho}_{AC} = \vert{}\psi\rangle\langle\psi\vert{}_A \otimes \sigma_C$,
for some positive semidefinite $\sigma_C$. Similarly applying the marginal preservation property for Alice's subsystem, we obtain that $\text{tr}_A[\tilde{\rho}_{AC}] = \sigma_C = \vert{}\phi\rangle\langle\phi\vert{}_C.$
We therefore see that the map returns the pure product state as
$$\Phi_y^\dagger(\vert{}\psi\rangle\langle\psi\vert{}_A \otimes \vert{}\phi\rangle\langle\phi\vert{}_C) = \vert{}\psi\rangle\langle\psi\vert{}_A \otimes \vert{}\phi\rangle\langle\phi\vert{}_C.$$

Now, recall that the set of all pure product states forms a basis that spans the entire  space $\mathcal{L}(\mathcal{H}_A \otimes \mathcal{H}_C)$. Here we explicitly use the assumption of finite-dimensional Hilbert spaces to guarantee that pure-state projectors ($|\psi \rangle \langle \psi|$) span the real space of Hermitian operators on $\mathcal{H}_A$ and likewise on $\mathcal{H}_C$. Consequently the set of all operators of the form $X_A \otimes Y_C$ with $X_A = X_A^{\dagger}$ and $Y_C = Y_C^{\dagger}$ is spanned by elements of the form $|\psi \rangle \langle \psi | \otimes |\phi \rangle \langle \phi|$. Since the linear map $\Phi_y^\dagger$ acts as the identity on a complete spanning set, it follows by linearity that it must act as the identity on all operators, giving :$$\Phi_y^\dagger \equiv \mathbb{I}_{AC}.$$
Substituting this back into the probability distribution yields:
$$P(a,b,c\vert{}x,y,z) = \text{tr}[\rho_{ABC}(M_{a}^{x}\otimes M_{c}^{z}\otimes M_{b}^{y})],$$ 
the usual Born rule of quantum theory.
\end{proof}



As we have seen, Thm. \ref{thm:jammingmap-pos} rules out any family of jamming maps $\{\Phi_y\}$ that preserve positivity on every quantum state. We are therefore left with the weaker requirement that the maps $\Phi_y$ (more precisely their adjoints) need only preserve positivity on the particular state $\rho_{AC}$ shared by Alice and Charlie. We now show that even this state-dependent version is unphysical - it reintroduces hidden superluminal signaling, but via a different mechanism analogous to that originally identified by Gisin \cite{Gisin90, SBG01} for nonlinear quantum mechanics.


Let us make this precise by presenting the physical argument here, with a formal Thm. \ref{thm:jamq-signal-app} and proof provided in App. \ref{sec:JammingQuantum-signal}. Accordingly, suppose that for every state $\rho_{AC}$ of Alice and Charlie's systems, and every input $y$, there exists a linear, trace-preserving, locally identity-preserving map $\Phi_y^{\rho_{AC}}$ such that $\left(\Phi_y^{\rho_{AC}}\right)^{\dagger}(\rho_{AC}) \geq 0$. Now suppose that the state $\rho_{AC}$ admits two different ensemble decompositions: 
\begin{eqnarray}
\rho_{AC} &=& p {\rho^0}_{AC} + (1-p) \rho^1_{AC}, \nonumber \\
 &=& q \sigma^0_{AC} + (1- q) \sigma^1_{AC},
\end{eqnarray}
for some $p, q \in (0,1)$. A remote party David (D) who holds a purification of $\rho_{AC}$ can choose which ensemble is realized by performing a measurement indexed by $w \in \{0,1\}$ on the purifying system. When the state-dependent map is applied after the ensemble is prepared, the final states seen by Alice and Charlie are $p \left(\Phi_y^{\rho^0_{AC}}\right)^{\dagger}(\rho^{0}_{AC}) + (1-p) \left(\Phi_y^{\rho^1_{AC}}\right)^{\dagger}(\rho^1_{AC})$ when $w=0$, versus $q \left(\Phi_y^{\sigma^0_{AC}}\right)^{\dagger}(\sigma^0_{AC}) + (1-q) \left(\Phi_y^{\sigma^1_{AC}}\right)^{\dagger}(\sigma^1_{AC})$ when $w=1$. In case these two mixtures are different, Alice and Charlie can distinguish which ensemble was prepared and thereby learn David's measurement choice $w$. Since David can be space-like separated from Alice and Charlie, this constitutes superluminal signaling. 

The only way to prevent such signaling is to demand that the final state seen by Alice and Charlie depend only on the average state $\rho_{AC}$, and not on the specific ensemble used to prepare it. In other words, for every input $y$ and every convex decomposition of $\rho_{AC}$ as $\rho_{AC} = p \rho^0_{AC} + (1-p) \rho^1_{AC}$, one must have
\begin{eqnarray}
\label{eq:affine}
    &&\left(\Phi_y^{p \rho^0_{AC} + (1-p) \rho^1_{AC}}\right)^{\dagger}\left(p \rho^0_{AC} + (1-p) \rho^1_{AC} \right) \nonumber \\
    &&= p \left(\Phi_y^{\rho^0_{AC}} \right)^{\dagger}(\rho^0_{AC}) + (1-p) \left(\Phi_y^{\rho^1_{AC}}\right)^{\dagger}(\rho^1_{AC}).
\end{eqnarray}
Note that the same equality must hold for every ensemble. Eq. \ref{eq:affine}  is a strong consistency condition that relates the maps belonging to different states. Combined with the already assumed properties of $\Phi_y$ - namely, linearity, trace preservation, and the requirement that the map act as the identity on every local subsystem - this condition forces a drastic simplification. As shown formally in Lemma \ref{lem:Gisin} and Thm. \ref{thm:jamq-signal-app} in Appendix \ref{sec:JammingQuantum-signal}, the effective evolution is then necessarily the identity on every input state, i.e.,
\begin{equation}
    \left(\Phi_y^{\rho_{AC}}\right)^{\dagger}(\rho_{AC}) = \rho_{AC},
\end{equation}
for every $\rho_{AC}$ and every $y$, and no nontrivial jamming is possible. Therefore, we conclude that no nontrivial jamming extension of quantum correlations is possible - any attempt to introduce nontrivial jamming maps in either state-independent or state-dependent fashion leads to hidden superluminal signaling.

\textit{Implications for security of Device-Independent Protocols against Relativistic adversaries.-}
In a device-independent (DI) protocol, no trust is placed on an internal model for the devices, instead the devices receive inputs and return outputs and security is inferred from the observed conditional distributions $\{P(a,b,c|x,y,z)\}$. DI protocols for cryptographic primitives such as key distribution \cite{BHK05}, random number generation \cite{Colbeck07, Brandao16, RLW25}, bit commitment \cite{FF15} and secret sharing \cite{MBNC20} have been constructed, and the ultimate notion of security considered thus far has been that of a no-signaling (NS) adversary, where the set of behaviors is constrained to lie in the $\mathcal{NS}$ polytope. On the other hand, as we have seen, relativistic causality imposes only those marginal constraints that are compatible with the spacetime locations of the corresponding measurement events - such additional freedom can be exploited by a relativistic adversary. Some implications on the monogamy of nonlocality and DI protocols for random number generation and key distribution were considered in \cite{HR19, SKGH+21}. On the other hand, the question was left open for a different class of primitives termed mistrustful cryptographic protocols in which one or more protocol participants themselves may be dishonest. We illustrate this point with two mistrustful primitives which have been shown to be secure against no-signaling adversaries. 

Consider first the three-prover bit-commitment protocol of Fehr and Fillinger \cite{FF15}. Here, a classical verifier interacts with three separated devices $A, B$ and $C$. The committing party coordinates these devices and uses them to commit to a bit $s \in \{0,1\}$. During the commitment phase, the verifier sends questions to the devices and records their answers. During the opening phase, the committing party announces $s$ together with auxiliary data, and the verifier sends additional test questions to randomly selected devices - the opening is accepted only if the answers satisfy the prescribed consistency checks. The protocol is hiding when the transcript does not reveal $s$, while it is binding when a dishonest committing party cannot produce accepting openings for both $s=0$ and $s=1$ with high probability. It is important to note that the three-prover structure is crucial here - separation prevents the devices from coordinating their answers after receiving the verifier's questions. In a jamming configuration however, one of the devices can influence the joint correlations of the other two without changing either of their marginals. As we show rigorously in Appendix \ref{sec:Jam-BitCommitment}, this extra freedom can be used to coordinate valid opening strings for both committed bits, thereby allowing the committing party to bypass the binding requirement. Specifically, the jamming device's input selects which pair of strings is jointly realized while the single-device statistics and hence the verifier's observed local tests remain unchanged. This allows the adversary to achieve the sum of success probabilities for opening to $0$ and $1$ in $\mathcal{RC}$ theories as 
\begin{equation}
    p_0^{RC} + p_1^{RC} \geq 2 - 2 \epsilon - 2 \delta,
\end{equation}
where $\epsilon$ and $\delta$ quantify the protocol's hiding and soundness errors. In the ideal limit with $\epsilon, \delta \to 0$, the right-hand side approaches $2$ violating the security bound proven against no-signaling adversaries of $p_0^{NS} + p_1^{NS} \leq 1+ \eta$ with $\eta < 1$ in \cite{FF15}. 

An alternative jamming mechanism also breaks DI secret sharing \cite{MBNC20}. In the secret sharing protocol, the honest devices produce correlations certifying that a designated group must collaborate in order to reconstruct a secret, while an external devices should have negligible guessing probability of the secret. Here, an RC jammer is able to alter the relevant joint Alice-Charlie correlations even while maintaining the tested statistics, and is able to supply an output that determines the secret perfectly leading to a guessing probability $P_{guess} = 1$. These examples go beyond the lack of monogamy in RC theories noted in \cite{HR19, SKGH+21} and show that security against no-signaling devices is not sufficient to imply security against relativistic adversaries. Notably, the considered behavior must be specified together with the spacetime locations of its measurement events, the input-output statistics alone are insufficient. 

\textit{Discussion.-}
Relativistic causality in jamming configurations allow for correlations beyond the standard no-signaling set termed $\mathcal{RC}$ correlations. We proved that the analogous extension that keeps a quantum state and allows such jamming necessarily makes the effective dynamics state-dependent. Under the requirements of ensemble equivalence and no-signaling for all preparations, this state dependence reintroduces superluminal signaling. Consequently the remaining consistent possibilities are ordinary quantum theory (without any jamming) or a fully post-quantum $\mathcal{RC}$ theory built on arbitrary Hermitian operators. The latter is completely characterized by a simple operator representation that generalizes the no-signaling framework of Acin et al. \cite{Acin10}.

The relativistic causality framework makes us take geometry-dependent causal constraints seriously. The full no-signaling assumption is an over-idealisation once the actual spacetime locations of the parties are taken into account. Specifically, in DI protocols aiming to be secure against relativistic adversaries, we have shown that the relativistic causal correlations enable perfect attacks on DI bit commitment and secret sharing that bypasses all no-signaling attacks. It reinforces the idea \cite{HR19, SKGH+21} that the standard no-signaling adversary model is insufficient once the causal structure of the experiment is taken into account, and the $\mathcal{RC}$ set of Thm. \ref{thm:RC-rep} is the correct replacement. 

Other approaches have studied the physicality of jamming mechanisms. For instance, \cite{Weilenmann25} studies the constraints on relativistic causal correlations via entropic monogamy and concludes (by alternative means) that it is likely that physical jamming mechanisms must signal (see also \cite{EMHRH25} for a discussion and counter-arguments). Our second result is a complementary no-go result  - any attempt to realize jamming correlations while retaining a quantum state forces state-dependent dynamics that reintroduce superluminal signaling under standard operational assumptions. It has also been noted in axiomatic reconstructions of quantum theory \cite{CDP11} that combining standard density matrices with non-causal evolutions leads to paradoxes. In this context, our second result might be viewed as a rigorous no-go theorem proving exactly why this is the case for relativistically causal jamming scenarios. In general, a theory in which effects such as jamming occur by point-to-region influences, may eventually turn out to be incompatible with stability, renormalisability or a continuum limit, but that incompatibility has to be clearly demonstrated and understood rather than assumed a priori. As such, it is important to understand the physical and information-theoretical principles that rule out jamming mechanisms and their consequences for security of cryptographic protocols that aim for security against relativistic adversaries.

\textit{Acknowledgments.-}
We acknowledge support from the General Research Fund (GRF) Grant No.\ 17307925, and the Research Impact Fund (RIF) Grant No.\ R7035-21.	

\onecolumngrid
\appendix
\counterwithin*{theorem}{section}
\counterwithin*{lemma}{section}
\counterwithin*{definition}{section}
\counterwithin*{example}{section}
\counterwithin*{assumption}{section}

\setcounter{assumption}{0} 
\renewcommand{\theassumption}{\thesection.\arabic{assumption}} 
\setcounter{example}{0} 
\renewcommand{\theexample}{\thesection.\arabic{example}} 
\setcounter{definition}{0} 
\renewcommand{\thedefinition}{\thesection.\arabic{definition}} 
\setcounter{theorem}{0} 
\renewcommand{\thetheorem}{\thesection.\arabic{theorem}} 
\setcounter{lemma}{0} 
\renewcommand{\thelemma}{\thesection.\arabic{lemma}} 
\setcounter{fact}{0} 
\renewcommand{\thefact}{\thesection.\arabic{fact}} 
\newpage

\section{Unified Operator Framework for Relativistic Causal Correlations}
\label{sec:RC}

\begin{theorem}
\label{thm:RC-rep-app}
An arbitrary tripartite behavior $\{P(a,b,c\vert{}x,y,z)\}$ satisfies the Relativistic Causality $\mathcal{RC}$ constraints for a jamming space-time configuration (where the intersection of Alice and Charlie's future light cones is contained within Bob's future light cone) if and only if there exist:  
\begin{itemize}
    \item Local Hilbert spaces $\mathcal{H}_A, \mathcal{H}_B, \mathcal{H}_C$ and sets of local quantum measurements $\{M_a^x\}, \{M_b^y\}, \{M_c^z\}$ (that sum to their respective local identities, e.g., $\sum_b M_b^y = \mathbb{I}_B$), 
    \item  a Hermitian operator $\Omega_{ABC} \in \mathcal{L}(\mathcal{H}_A \otimes \mathcal{H}_B \otimes \mathcal{H}_C)$ of unit trace,
    \item a family of linear, trace-preserving maps $\Phi_y : \mathcal{L}(\mathcal{H}_A \otimes \mathcal{H}_C) \to \mathcal{L}(\mathcal{H}_A \otimes \mathcal{H}_C)$, parameterized by Bob's setting $y$,
\end{itemize}  
such that:
\begin{equation}
    \label{eq:gen-trace-rule-app}
    P(a,b,c\vert{}x,y,z) = \text{tr}\Big[ \Omega_{ABC} \big( \Phi_y(M_a^x \otimes M_c^z) \otimes M_b^y \big) \Big],
\end{equation}
for all $a,b,c,x,y,z$.
And the map $\Phi_y$ for any $y$ is locally identity-preserving, i.e., its action on any local observable $X_A \in \mathcal{L}(\mathcal{H}_A)$ and $X_C \in \mathcal{L}(\mathcal{H}_C)$ satisfies:
\begin{eqnarray}
\label{eq:Phiy-RC}
\Phi_y(X_A \otimes \mathbb{I}_C) &=& X_A \otimes \mathbb{I}_C, \nonumber \\ \Phi_y(\mathbb{I}_A \otimes X_C) &=& \mathbb{I}_A \otimes X_C.
\end{eqnarray}
\end{theorem}

\begin{proof}
Sufficiency (If). We first prove that if the behavior is generated as per \eqref{eq:gen-trace-rule-app}, it satisfies all the RC constraints required to prevent causal loops \cite{HR19}.  

We can readily check that:
\begin{eqnarray}
 \sum_a P(a,b,c|x,y,z) &=& \sum_a \text{tr}\Big( \Omega_{ABC} \big( \Phi_y(M_a^x \otimes M_c^z) \otimes M_b^y \big) \Big) \nonumber \\
 &=& \text{tr}\Big( \Omega_{ABC} \big( \Phi_y(\mathbb{I}_A \otimes M_c^z) \otimes M_b^y \big) \Big) \nonumber \\
 &=& \text{tr}\Big( \Omega_{ABC} \big( \mathbb{I}_A \otimes M_c^z \otimes M_b^y \big) \Big) = P(b,c|y,z).
\end{eqnarray}
Here we have used that $\sum_a M_a^x = \mathbb{I}_A$ and that $\Phi_y(\mathbb{I}_A \otimes M_c^z) = \mathbb{I}_A \otimes M_c^z$. The result is manifestly independent of $x$, precluding Alice from signaling to Bob and Charlie. By symmetry, summing over Charlie's outcomes yields $P(a,b\vert{}x,y)$, which is independent of $z$. 

Similarly, we check that 
\begin{eqnarray}
    \sum_{b,c} P(a,b,c|x,y,z) &=& \text{tr}\Big( \Omega_{ABC} \big( \Phi_y(M_a^x \otimes \mathbb{I}_C) \otimes \mathbb{I}_B \big) \Big) \nonumber \\
    &=&\text{tr}\Big( \Omega_{ABC} \big( M_a^x \otimes \mathbb{I}_C \otimes \mathbb{I}_B \big) \Big) = P(a|x).
\end{eqnarray}
Here we have again used that $\Phi_y(M_a^x \otimes \mathbb{I}_C) = M_a^x \otimes \mathbb{I}_C$. We verify that this marginal is independent of Bob and Charlie's settings $y$ and $z$. By symmetry, Charlie's marginal distribution is independent of Alice and Bob's settings $x$ and $y$. Furthermore,  
\begin{eqnarray}
    \sum_{a,c} P(a,b,c|x,y,z) &=& \sum_{a,c} \text{tr}\Big( \Omega_{ABC} \big( \Phi_y(M_a^x \otimes M_c^z) \otimes M_b^y \big) \Big) \nonumber \\
    &=&\text{tr}\Big( \Omega_{ABC} \big( \mathbb{I}_A \otimes \mathbb{I}_C \otimes M_b^y \big) \Big) = P(b|y). 
\end{eqnarray}
Here again we have used that $\Phi_y(M_a^x \otimes \mathbb{I}_C) = M_a^x \otimes \mathbb{I}_C$. We verify that Bob's marginal distribution is independent of Alice and Charlie's settings $x$ and $z$. 

Now, we have that 
\begin{eqnarray}
    P(a,c\vert{}x,y,z) &=& \text{tr}\Big( \Omega_{ABC} \big( \Phi_y(M_a^x \otimes M_c^z) \otimes \mathbb{I}_B \big) \Big).
\end{eqnarray}
Here $\Phi_y$ maps $M_a^x \otimes M_c^z$ to a $y$-dependent operator, so that the joint AC correlation depends on $y$, as per the jamming condition. We conclude that the generated behavior satisfies all the RC constraints and at the same time allows for a jamming influence from Bob to the joint correlations between Alice and Charlie's systems. 

Necessity (Only If). We now proceed to show that any behavior $P(a,b,c\vert{}x,y,z)$ satisfying the RC constraints can be written in terms of \eqref{eq:gen-trace-rule-app} with a suitable operator $\Omega_{ABC}$ and a linear trace-preserving map $\Phi_y$. We follow the proof of \cite{Acin10} here,  we first choose local Hilbert spaces $\mathcal{H}_A, \mathcal{H}_B, \mathcal{H}_C$ and for each party we define a set of local quantum measurements (POVMs) such that for all settings, the operators excluding the final outcome—e.g., $\{M_a^x\}_{a<r_A}$—along with the local identity $\mathbb{I}_A$, form a linearly independent set. The final outcome is determined by $M_{r_A}^x = \mathbb{I}_A - \sum_{a<r_A} M_a^x$.  Now, since these primal operators are linearly independent, there exists a uniquely defined dual basis for each party \cite{Acin10}. For Alice, this is $\{\tilde{\Lambda}_A, \tilde{M}_a^x\}$, which satisfies the orthogonality conditions:
  $$\text{tr}(\tilde{M}_a^x M_{a'}^{x'}) = \delta_{aa'}\delta_{xx'}, \quad  \quad \text{tr}(\tilde{\Lambda}_A \mathbb{I}_A) = 1,$$$$\text{tr}(\tilde{M}_a^x \mathbb{I}_A) = 0, \quad  \quad \text{tr}(\tilde{\Lambda}_A M_a^x) = 0.$$ 
We define similar primal/dual bases for Bob ($\{\mathbb{I}_B, M_b^y\}$, $\{\tilde{\Lambda}_B, \tilde{M}_b^y\}$) and Charlie ($\{\mathbb{I}_C, M_c^z\}$, $\{\tilde{\Lambda}_C, \tilde{M}_c^z\}$).

Now define a fully non-signaling reference behavior $Q(a,b,c\vert{}x,y,z)$ that shares all the 1-body and 2-body marginals of $P(a,b,c\vert{}x,y,z)$ that involve Bob. Specifically consider the behavior 
\begin{equation}
    Q(a,b,c\vert{}x,y,z) = \frac{P(a,b\vert{}x,y)P(b,c\vert{}y,z)}{ P(b\vert{}y)},
\end{equation} 
defined for all $b, y$ such that $P(b \vert{}y) \neq 0$ and utilized in the proofs in \cite{SKGH+21, Fine82}. Since $Q$ is explicitly no-signaling, by the result of \cite{Acin10} we can map it into an operator $\Omega_{ABC}$ of unit trace using the expansion:
\begin{eqnarray}
\label{eq:Omega-Acin}
\Omega_{ABC} &=& \sum_{a,b,c,x,y,z} Q(a,b,c\vert{}x,y,z) \; \tilde{M}_a^x \otimes \tilde{M}_b^y \otimes \tilde{M}_c^z + \sum_{a,b,x,y} Q(a,b\vert{}x,y) \; \tilde{M}_a^x \otimes \tilde{M}_b^y \otimes \tilde{\Lambda}_C + \dots + \tilde{\Lambda}_A \otimes \tilde{\Lambda}_B \otimes \tilde{\Lambda}_C. \nonumber \\
\end{eqnarray}
We check that by the orthogonality of the dual basis, we have that $\text{tr}(\Omega_{ABC} \left[M_a^x \otimes M_b^y \otimes M_c^z)) = Q(a,b,c\vert{}x,y,z \right]$. 

To account for the additional jamming influence, we define the linear map $\Phi_y$ in terms of a perturbation as:
$$\Phi_y(M_a^x \otimes M_c^z) = M_a^x \otimes M_c^z + \Delta_{ac}^{xz}(y).$$
We require this perturbation operator to obey the constraint:
\begin{eqnarray}
    \text{tr}\Big[ \Omega_{ABC} \big( \Phi_y(M_a^x \otimes M_c^z) \otimes M_b^y \big) \Big]  
    &=& \text{tr}\Big[ \Omega_{ABC} \big( (M_a^x \otimes M_c^z + \Delta_{ac}^{xz}(y)) \otimes M_b^y \big) \Big] \nonumber \\
    &=& P(a,b,c\vert{}x,y,z).
\end{eqnarray}
Since the unperturbed trace yields $Q$, we obtain the explicit constraint for the perturbation as:
\begin{eqnarray}
\label{eq:pert-const}
    \text{tr}\Big[ \Omega_{ABC} \big( \Delta_{ac}^{xz}(y) \otimes M_b^y \big) \Big] = P(a,b,c\vert{}x,y,z) - Q(a,b,c\vert{}x,y,z) \equiv D(a,b,c\vert{}x,y,z). \nonumber \\
\end{eqnarray}
Finally, to guarantee that $\Phi_y(X_A \otimes \mathbb{I}_C) = X_A \otimes \mathbb{I}_C$, the perturbation $\Delta_{ac}^{xz}(y)$ must vanish when traced against any local identity. We enforce this algebraically by expanding $\Delta$ in the non-identity primal matrices:
\begin{equation}
\label{eq:Delta}
\Delta_{ac}^{xz}(y) = \sum_{x', a'<r_A} \sum_{z', c'<r_C} \Gamma_{a'c'}^{x'z'}(a,c,x,z,y) \; M_{a'}^{x'} \otimes M_{c'}^{z'}.
\end{equation}
Since $M_{a'}^{x'}$ and $M_{c'}^{z'}$ are orthogonal to the dual identity elements $\tilde{\Lambda}_A$ and $\tilde{\Lambda}_C$, this guarantees $\Delta$ only acts on the bipartite correlation subspace, isolating Bob's jamming from Alice and Charlie's local marginals. 

We can now solve for the coefficients $\Gamma$, by substituting the primal expansion of $\Delta$ in \eqref{eq:Delta} into the constraint \eqref{eq:pert-const}. 
Since $\Omega_{ABC}$ in \eqref{eq:Omega-Acin} is written in the dual basis, tracing it with the primal matrices $M_{a'}^{x'} \otimes M_b^y \otimes M_{c'}^{z'}$ directly extracts the coefficients $Q(a',b,c'\vert{}x',y,z')$ and we obtain:
\begin{eqnarray}
\label{eq:coeffs}
    \sum_{a', c', x', z'} \Gamma_{a'c'}^{x'z'}(a,c,x,z,y) \cdot Q(a',b,c'\vert{}x',y,z') = D(a,b,c\vert{}x,y,z). 
\end{eqnarray}
This forms a system of linear equations. Specifically for a fixed $(a,c,x,y,z)$, Eq.\eqref{eq:coeffs} represents a system of linear equations evaluated across all of Bob's non-deterministic outcomes $b < r_B$. We write this as a matrix equation $\mathbf{Q} \; \vec{\Gamma} = \vec{D}$, where $\vec{D}$ is a vector of length $(r_B - 1)$, $\vec{\Gamma}$ is a vector of length $N = m_A(r_A - 1) \times m_C(r_C - 1)$, and $\mathbf{Q}$ is a matrix of dimensions $(r_B - 1) \times N$ where each row corresponds to a specific outcome $b$ and each column corresponds to an evaluation of the reference behavior $Q$ at the basis variables $(a', c', x', z')$.

To guarantee that an exact solution for the vector $\vec{\Gamma}$ exists, the matrix $Q$ must possess full row rank (rank $r_B - 1$). If $Q$ is built using only the physical measurement settings $x$ and $z$, this full rank is not guaranteed. To force full row rank without altering the physical behavior, the local measurement sets for Alice and Charlie are expanded to include $K = r_B - 1$ auxiliary dummy settings: $\{x_1^*, \dots, x_K^*\}$ and $\{z_1^*, \dots, z_K^*\}$ with binary outcomes $a', c' \in \{0,1\}$. Since these are mathematical constructs, their reference distribution $Q$ can be engineered arbitrarily, provided it remains non-signaling and preserves Bob's physical marginal $P(b\vert{}y)$. To be specific, we define these as follows.  For each dummy input pair corresponding to $k \in \{1, \dots, r_B - 1\}$, we define the joint reference probabilities as follows: $$Q(0, b, 0 \vert{} x^*_k, y, z^*_k) = \frac{1}{4}P(b\vert{}y) + \epsilon \cdot P(b\vert{}y) \delta_{bk},$$ $$Q(1, b, 0 \vert{} x^*_k, y, z^*_k) = \frac{1}{4}P(b\vert{}y) - \epsilon \cdot P(b\vert{}y) \delta_{bk}$$$$Q(0, b, 1 \vert{} x^*_k, y, z^*_k) = \frac{1}{4}P(b\vert{}y) - \epsilon \cdot P(b\vert{}y) \delta_{bk},$$ $$Q(1, b, 1 \vert{} x^*_k, y, z^*_k) = \frac{1}{4}P(b\vert{}y) + \epsilon \cdot P(b\vert{}y) \delta_{bk},$$
where $\epsilon > 0$ is a sufficiently small constant (e.g., $\epsilon \le 1/4$) ensuring all probabilities remain non-negative.

We verify that this block is fully non-signaling by computing: $$Q(0, b, c' \vert{} x^*_k, y, z^*_k) + Q(1, b, c' \vert{} x^*_k, y, z^*_k) = \frac{1}{2}P(b\vert{}y).$$ This is independent of Alice's setting $x^*_k$, meaning Alice cannot signal to Bob or Charlie. Similarly we compute $$Q(a', b, 0 \vert{} x^*_k, y, z^*_k) + Q(a', b, 1 \vert{} x^*_k, y, z^*_k) = \frac{1}{2}P(b\vert{}y).$$ This is independent of Charlie's setting $z^*_k$. Finally, Bob's isolated marginal (obtained by summing over $a'$ and $c'$) yields exactly $P(b\vert{}y)$.

We now isolate a specific submatrix of $\mathbf{Q}$ by looking at the columns corresponding to the outcomes $a'=0, c'=0$ evaluated at the correlated dummy settings $(x^*_k, z^*_k)$. Let column $k$ (for $k = 1, \dots, r_B - 1$) be the vector $\vec{v}_k$ representing the equation coefficients evaluated across all rows $b$: $$\vec{v}_k = \Big[ Q(0, b, 0 \vert{} x^*_k, y, z^*_k) \Big]_{b=1}^{r_B-1} = \frac{1}{4} \vec{P}(b\vert{}y) + \epsilon P(k\vert{}y) \hat{e}_k,$$
where $\hat{e}_k$ is the standard basis vector with a $1$ at position $k$ and $0$ elsewhere. The square $(r_B - 1) \times (r_B - 1)$ submatrix formed by these columns is:
\begin{eqnarray}
\mathbf{Q}_{\text{sub}}
    = \frac{1}{4} \vec{P}(b\vert{}y) \cdot \vec{1}^T + \epsilon \cdot \text{diag}(P(1\vert{}y), P(2\vert{}y), \dots, P(r_B-1\vert{}y)). 
\end{eqnarray}

This $\mathbf{Q}_{\text{sub}}$ is a diagonal matrix of positive entries (since $P(b\vert{}y) > 0$) combined with a rank-1 update (the outer product $\frac{1}{4}\vec{P}(b\vert{}y)\mathbf{1}^T$). By the Matrix Determinant Lemma, this guarantees that $Q_{sub}$ is invertible. Now since $\mathbf{Q}$ contains a fully invertible $(r_B - 1) \times (r_B - 1)$ submatrix, the entire matrix $\mathbf{Q}$ is  guaranteed to possess full row rank. By comparing the rank of the coefficient matrix and the augmented matrix, 
we see that the surjective linear map $\mathbf{Q} \; \vec{\Gamma} = \vec{D}$ will always yield at least one exact solution vector $\vec{\Gamma}$. The coefficients $\Gamma_{a'c'}^{x'z'}$ are obtained simply by computing the pseudo-inverse (or directly inverting the submatrix and setting all other coefficients to zero):
$$\vec{\Gamma}_{\text{sub}} = \mathbf{Q}_{\text{sub}}^{-1} \vec{D}.$$

Finally, we verify Bob's final outcome ($b=r_B$). Summing the constraint in \eqref{eq:pert-const} over all $b$ yields:
\begin{eqnarray}
    &&\sum_b \text{tr}\Big[ \Omega_{ABC} \big( \Delta_{ac}^{xz}(y) \otimes M_b^y \big) \Big] \nonumber \\
    &&=\text{tr}\Big[ \Omega_{ABC} \big( \Delta_{ac}^{xz}(y) \otimes \mathbb{I}_B \big) \Big] \nonumber \\
    &&= \sum_b D(a,b,c\vert{}x,y,z) = P(a,c\vert{}x,y,z) - Q(a,c\vert{}x,z). 
\end{eqnarray}
This maps Bob's input parameter $y$ onto the joint AC marginal, indicating the jamming.

\end{proof}

Remark that in the limit when $\Phi_y = \mathbb{I}_{AC}$ (the identity superoperator), the Thm. \ref{thm:RC-rep-app} recovers the framework of \cite{Acin10} for fully no-signalling correlations. 

As an illustrative example of Thm. \ref{thm:RC-rep-app}, we identify the operator $\Omega_{ABC}$, the measurements $\{M^x_a\}$, $\{M^y_b\}$ and $\{M^z_c\}$, and the maps $\Phi_y$ that produce the paradigmatic RC box given as Example 4 in \cite{HR19}. Specifically, this example considers a Bell scenario with three players Alice, Bob and Charlie in the jamming configuration. Each player performs one of two measurements with binary outputs resulting in the box $P(a,b,c\vert{}x,y,z)$. The box is characterised by deterministic output $b=0$ for both inputs $y=0, 1$ of Bob, with the marginal $P(a,c\vert{x}x,y,z)$ being a local box that returns uniformly random outcomes $a=c$ for all $x,z$ when Bob inputs $y=0$, and the marginal $P(a,c\vert{}x,y,z)$ being a Popescu-Rohrlich (PR) box \cite{PR94} that obeys $a \oplus c = x \cdot z$ when Bob inputs $y=1$. 

To build this behavior within the framework of Thm. \ref{thm:RC-rep-app}, we formally present the following example.

\begin{example}
Let the Hilbert spaces for Alice, Bob, and Charlie correspond to qubits: $\mathcal{H}_A = \mathcal{H}_B = \mathcal{H}_C = \mathbb{C}^2$.  As the operator $\Omega_{ABC}$, consider 
\begin{eqnarray}
    \Omega_{ABC} = \vert{}\Phi^+\rangle\langle\Phi^+\vert{}_{AC} \otimes \vert{}0\rangle\langle0\vert{}_B.
\end{eqnarray}
Specifically, the $AC$ state is 
$$\rho_{AC} = \frac{1}{4}(\mathbb{I} \otimes \mathbb{I} + \sigma_x \otimes \sigma_x - \sigma_y \otimes \sigma_y + \sigma_z \otimes \sigma_z).$$
As the local measurement operators, we choose the following. Alice measures in the standard Pauli $Z$ and $X$ bases: $M_a^0 = \frac{1}{2}(\mathbb{I}_A + (-1)^a \sigma_z)$, $M_a^1 = \frac{1}{2}(\mathbb{I}_A + (-1)^a \sigma_x).$ Charlie measures in a rotated bases given by  $M_c^0 = \frac{1}{2}(\mathbb{I}_C + (-1)^c \frac{\sigma_z + \sigma_x}{\sqrt{2}})$, $M_c^1 = \frac{1}{2}(\mathbb{I}_C + (-1)^c \frac{\sigma_z - \sigma_x}{\sqrt{2}})$. Since Bob's output is deterministically $b=0$ for all inputs, his measurement operators are trivial for both settings $y=0$ and $y=1$, namely $M_0^y = \mathbb{I}_B$, $M_1^y = 0$. 

The jamming maps $\Phi_y: \mathcal{L}(\mathcal{H}_A \otimes \mathcal{H}_C) \rightarrow \mathcal{L}(\mathcal{H}_A \otimes \mathcal{H}_C)$ are required to be linear, trace-preserving, and locally identity-preserving. We define these by their action on the basis of joint Pauli operators $\sigma_i \otimes \sigma_j$. We set when $y=0$:
\begin{eqnarray}
    \Phi_0(\sigma_z \otimes \sigma_z) &=& \sqrt{2} \sigma_z \otimes \sigma_z, \nonumber \\ 
    \Phi_0(\sigma_x \otimes \sigma_z) &=& \sqrt{2} \sigma_x \otimes \sigma_x, \nonumber \\
    \Phi_0(\sigma_y \otimes \sigma_y) &=& 0.
\end{eqnarray}
With all other Pauli basis operators mapping to themselves (or zero as the degree of freedom involved in the characterization). When $y=1$, we set:
\begin{eqnarray}
    \Phi_1(\sigma_z \otimes \sigma_z) &=& \sqrt{2} \sigma_z \otimes \sigma_z, \nonumber \\
    \Phi_1(\sigma_x \otimes \sigma_x) &=& \sqrt{2} \sigma_x \otimes \sigma_x, \nonumber \\
    \Phi_1(\sigma_y \otimes \sigma_y) &=& 0.
\end{eqnarray}
Again, all other Pauli basis operators map to themselves or to zero. 

We can verify that the exact RC behavior is achieved by computing 
$$P(a,b,c\vert{}x,y,z) = \text{tr}[\Omega_{ABC} (\Phi_y(M_a^x \otimes M_c^z) \otimes M_b^y)].$$
Since $b=0$ is guaranteed, Bob's subspace trivially traces to 1. Using the cyclic property of the trace, we can apply the adjoint map $\Phi_y^\dagger$ to the $AC$ state instead:
$$P(a,0,c\vert{}x,y,z) = \text{tr}[\Phi_y^\dagger(\rho_{AC}) (M_a^x \otimes M_c^z)].$$ For $y=0$, we have that $\Phi_0^{\dagger}(\rho_{AC}) =: \tilde{\rho}_0$ is given by
$$\tilde{\rho}_0 = \frac{1}{4}(\mathbb{I} \otimes \mathbb{I} + \sqrt{2}\sigma_z \otimes \sigma_z + \sqrt{2}\sigma_x \otimes \sigma_z).$$
Evaluating the expectation values of Alice and Charlie's observables on this state yields exactly $1$ for all four combinations of $x, z \in \{0, 1\}$, so that $a=c$ with probability 1, reproducing the local box from Example 4. For $y=1$, we have that $\Phi_1^{\dagger}(\rho_{AC}) =: \tilde{\rho}_1$ is given by 
$$\tilde{\rho}_1 = \frac{1}{4}(\mathbb{I} \otimes \mathbb{I} + \sqrt{2}\sigma_z \otimes \sigma_z + \sqrt{2}\sigma_x \otimes \sigma_x).$$
Evaluating the expectation values on this state yields $1$ for $(x,z) \in \{(0,0), (0,1), (1,0)\}$ and $-1$ for $(x,z) = (1,1)$. This reproduces the $a \oplus c = x \cdot z$ PR box correlations required by Example 4 from \cite{HR19}.
\end{example}

The above theorem applies to the paradigmatic three-party jamming configuration highlighted in \cite{GPR96, HR19, VC22}. The definition of relativistic causal correlations readily generalizes to the $n$-party configuration. 

\begin{definition}(The $n$-Party $\mathcal{RC}$ Set).
\label{def:n-party-RC}
Let $N=\{1, 2, \dots, n\}$ be a set of $n$ space-like separated parties in a spacetime configuration $\mathcal{M}$. Let $p_i$ and $q_i$ denote the input and output events for party $i\in N$, respectively.
For any subset of parties $S \subseteq N$, the Jamming Set $\mathcal{J}_\mathcal{M}(S)\subseteq N\setminus S$ is defined as the set of parties $j$ satisfying: $\bigcap_{i\in S} J^+(q_i)\subseteq J^+(p_j)$. A behavior $P(\vec{a}\vert{}\vec{x})$ belongs to $\mathcal{RC}$ if and only if there exist:
\begin{itemize}
\item Finite-dimensional Hilbert spaces $\{\mathcal{H}_i\}_{i\in N}$, with $\mathcal{H}_N=\bigotimes_{i\in N}\mathcal{H}_i$.

\item For each $i\in N$ and input $x_i$, a set of POVM elements $\{M_{a_i}^{x_i}\}\subset\mathcal{L}(\mathcal{H}_i)$ satisfying $M_{a_i}^{x_i}\geq 0$ and $\sum_{a_i} M_{a_i}^{x_i}=I_{\mathcal{H}_i}$.

\item A Hermitian operator $\Omega\in\mathcal{L}(\mathcal{H}_N)$ of unit trace.
\end{itemize}
Further, there exists a hierarchy of linear, trace-preserving maps defined for all subsets $S\subseteq N$: $\Phi^{(S)}_{\vec{u}}:\mathcal{L}(\mathcal{H}_S)\to\mathcal{L}(\mathcal{H}_S)$
where the parameter $\vec{u}=\vec{x}_{S\cup\mathcal{J}_\mathcal{M}(S)}$ depends only on the inputs of the parties in $S$ and their valid jammers. This hierarchy satisfies the following three axioms:

\begin{enumerate}
\item For every proper subset inclusion $S\subsetneq T\subseteq N$ and every operator $X_S\in\mathcal{L}(\mathcal{H}_S)$, the map acting on $T$ must perfectly intertwine with the map acting on $S$. As an equality of operators on $\mathcal{H}_T$, it must hold that:
\begin{equation}
    \Phi^{(T)}_{\vec{x}_{T\cup\mathcal{J}_\mathcal{M}(T)}}\bigl(X_S\otimes I_{T\setminus S}\bigr) = \Phi^{(S)}_{\vec{x}_{S\cup\mathcal{J}_\mathcal{M}(S)}}(X_S)\otimes I_{T\setminus S}.
\end{equation}

\item The map must fix all single-party operators that are not acting as jammers for the rest of the subset. For every $S\subseteq N$, every $i\in S$, and every $X_i\in\mathcal{L}(\mathcal{H}_i)$, if $i\notin\mathcal{J}_\mathcal{M}(S\setminus\{i\})$, then: 
\begin{equation}
    \Phi^{(S)}_{\vec{x}_{S\cup\mathcal{J}_\mathcal{M}(S)}}\bigl(X_i\otimes I_{S\setminus\{i\}}\bigr) = X_i\otimes I_{S\setminus\{i\}}.
\end{equation}

\item If the spacetime configuration permits no jamming at all ($\mathcal{J}_\mathcal{M}(S)=\emptyset$ for all $S$), then the global map reduces to the identity: 
\begin{equation}
    \Phi^{(N)}_{\vec{x}} = \operatorname{id}_{\mathcal{L}(\mathcal{H}_N)} \quad \forall\vec{x}.
\end{equation}
\end{enumerate}
The $n$-party relativistic causal correlations are written as:
\begin{equation}
    P(\vec{a}\vert{}\vec{x}) = \operatorname{tr}\Bigl(\Omega\cdot\Phi^{(N)}_{\vec{x}}\Bigl(\bigotimes_{i=1}^n M_{a_i}^{x_i}\Bigr)\Bigr).
\end{equation}
\end{definition}

In simple terms, in an $n$-party configuration, the first axiom on the jamming maps ensures that the correlations of a smaller subgroup are consistent and the jamming effects on that subgroup are self-contained. That is, it guarantees that the subgroup's measurement operators only depend upon the inputs within the subgroup and the specific parties explicitly permitted by the spacetime configuration to jam them. 
The second axiom guarantees that the measurement operators of any subset of parties is untouched by the jamming map, unless the spacetime configuration is such that an outside party is permitted to jam them. Furthermore, even an active jammer's own local measurement operator is unaffected by their jamming map. The third axiom simply states that if the spacetime configuration forbids jamming entirely, then the jamming maps reduce to the identity and the system behaves exactly like standard quantum mechanics. 

\section{Jamming extension of Quantum Correlations leads to Hidden Superluminal Signaling}
\label{sec:JammingQuantum-signal}

In this section, we examine the physicality of the intermediate set of quantum-like correlations provided in the main text and termed $\mathcal{RCQ}$. We recall the definition here for convenience, noting that the distinction with the general $\mathcal{RC}$ correlations of the previous section arises here from the requirement that the state be a positive semidefinite Hermitian operator $\rho_{ABC}$ of unit trace. Again, we stick to the three-party picture here for convenience, the general $n$-party case following immediately from Def. \ref{def:n-party-RC} by imposing the requirement of positive semidefiniteness for the state. As in the rest of the paper, in the following definition, $\mathcal{L}(\mathcal{H})$ denotes the space of linear operators on the Hilbert space $\mathcal{H}$, $\text{Herm}(\mathcal{H})$ denotes the space of all Hermitian (self-adjoint) operators acting on the Hilbert space $\mathcal{H}$, and $\mathcal{D}(\mathcal{H})$ denotes the space of density operators on the Hilbert space $\mathcal{H}$.

\begin{definition}
\label{def:RCQ-app}
    A tripartite behavior $P(a,b,c|x,y,z)$ belongs to the set of relativistically-causal quantum ($\mathcal{RCQ}$) correlations for a jamming space-time configuration if and only if there exist:
    \begin{itemize}
      \item local finite-dimensional Hilbert spaces $\mathcal{H}_A, \mathcal{H}_B, \mathcal{H}_C$ and sets of local quantum measurements $\{M_a^x\}, \{M_b^y\}, \{M_c^z\}$ (that sum to their respective local identities, e.g., $\sum_b M_b^y = \mathbb{I}_B$), 
    \item  a positive semidefinite Hermitian operator $\rho_{ABC} \geq 0$ in $\mathcal{L}(\mathcal{H}_A \otimes \mathcal{H}_B \otimes \mathcal{H}_C)$ of unit trace ($\text{tr}(\rho_{ABC}) = 1)$,
    \item a family of linear maps $\Phi_y : \mathcal{L}(\mathcal{H}_A \otimes \mathcal{H}_C) \to \mathcal{L}(\mathcal{H}_A \otimes \mathcal{H}_C)$, parameterized by Bob's setting $y$, that are trace-preserving $\Phi_y^{\dagger}(\mathbb{I}_A \otimes \mathbb{I}_C) = \mathbb{I}_A \otimes \mathbb{I}_C$, and preserve local identities, i.e., their action on any local observable $X_A \in \mathcal{L}(\mathcal{H}_A)$ and $X_C \in \mathcal{L}(\mathcal{H}_C)$ satisfies:
$$\Phi_y(X_A \otimes \mathbb{I}_C) = X_A \otimes \mathbb{I}_C,$$
$$\Phi_y(\mathbb{I}_A \otimes X_C) = \mathbb{I}_A \otimes X_C,$$ and furthermore, the maps $\Phi_y$ preserve positivity of the initial state, i.e., $\Phi_y^{\dagger}(\rho_{AC}) \succeq 0$,
\end{itemize}  
such that:
\begin{equation}
    \label{eq:gen-trace-rule}
    P(a,b,c\vert{}x,y,z) = \text{tr}\Big[ \rho_{ABC} \big( \Phi_y(M_a^x \otimes M_c^z) \otimes M_b^y \big) \Big].
\end{equation}
\end{definition}

The set $\mathcal{RCQ}$ is a strict superset of the usual quantum set in spacetime configurations which permit jamming. As a paradigmatic example, we present the following. 
\begin{example}
Suppose that Alice, Bob and Charlie in the jamming configuration share a 3-qubit state   
$$\rho_{ABC} = \vert{}\Phi^+\rangle\langle\Phi^+\vert{}_{AC} \otimes \vert{}0\rangle\langle0\vert{}_B,$$
where $\vert{}\Phi^+\rangle = \frac{\vert{}00\rangle + \vert{}11\rangle}{\sqrt{2}}$,
and perform standard dichotomic projective measurements given as follows. 
\begin{itemize}
    \item Alice ($x \in \{0, 1\}$): $M_{a}^{x} = \frac{\mathbb{I} + (-1)^a A_x}{2}$, where $A_0 = \sigma_z$ and $A_1 = \sigma_x$,
    \item Bob ($y \in \{0, 1\}$): $M_{b}^{y} = \frac{\mathbb{I} + (-1)^b B_y}{2}$, where $B_0 = \sigma_z$ and $B_1 = \sigma_x$,
    \item Charlie ($z \in \{0, 1\}$): $M_{c}^{z} = \frac{\mathbb{I} + (-1)^c C_z}{2}$, where $C_0 = \frac{\sigma_z + \sigma_x}{\sqrt{2}}$ and $C_1 = \frac{\sigma_z - \sigma_x}{\sqrt{2}}$.
\end{itemize}
The jamming maps $\Phi_y$ are given as follows. Let $\sigma_0 \equiv \mathbb{I}$, and let $\{\sigma_\mu \otimes \sigma_\nu\}_{\mu,\nu=0}^3$ be the orthogonal Pauli basis. For $y = 0$, Bob chooses not to jam $$\Phi_0 \equiv \mathbb{I}_{AC}.$$
For $y = 1$, we define $\Phi_1$ uniquely by its linear extension on the basis vectors as 
$$\Phi_1(\sigma_\mu \otimes \sigma_\nu) = \sigma_\mu \otimes \sigma_\nu \quad \text{if } \mu=0 \text{ or } \nu=0,$$
$$\Phi_1(\sigma_\mu \otimes \sigma_\nu) = 0 \quad \text{if } \mu \neq 0 \text{ and } \nu \neq 0.$$
Since $\Phi_1$ maps Hermitian basis operators to themselves or zero, it is self-adjoint ($\Phi_1^\dagger = \Phi_1$).
Furthermore, for any Alice observable $X_A = \sum_\mu c_\mu \sigma_\mu$, we evaluate:
$$\Phi_1(X_A \otimes \mathbb{I}) = \sum_\mu c_\mu \Phi_1(\sigma_\mu \otimes \sigma_0).$$
Because $\nu = 0$, the map preserves the term: $\sum_\mu c_\mu (\sigma_\mu \otimes \sigma_0) = X_A \otimes \mathbb{I}$.
Similar considerations hold for $\mathbb{I} \otimes X_C$. 
We compute the joint probability distributions of outcomes given inputs as
\begin{eqnarray}
    P(a,b,c\vert{}x,y,z) &=& \text{tr}[\rho_{ABC}(\Phi_y(M_a^x \otimes M_c^z) \otimes M_b^y)] \nonumber \\
    &=& \text{tr}[\Phi_y^\dagger(\rho_{AC}) (M_a^x \otimes M_c^z)] \cdot \text{tr}[\vert{}0\rangle\langle0\vert{}_B M_b^y].\nonumber \\
\end{eqnarray}
We evaluate the effective state $\tilde{\rho}_{AC\vert{}y} = \Phi_y^\dagger(\vert{}\Phi^+\rangle\langle\Phi^+\vert{})$.
For $y = 0$: $\tilde{\rho}_{AC\vert{}0} = \Phi_0(\vert{}\Phi^+\rangle\langle\Phi^+\vert{}) = \vert{}\Phi^+\rangle\langle\Phi^+\vert{}$. With this combination of state and measurements, one has $\langle CHSH \rangle_{AC|y=0} = 2 \sqrt{2}$ for $\langle CHSH \rangle_{AC|y=0} = \text{Tr}\left[\left( A_0 C_0 + A_0 C_1 + A_1 C_0 - A_1 C_1 \right){\rho_{AC|y=0}} \right]$. 
For $y = 1$, the map $\Phi_1^\dagger$ perfectly annihilates the correlation terms, leaving: $\tilde{\rho}_{AC\vert{}1} = \frac{1}{4}(\mathbb{I}\otimes\mathbb{I}) = \frac{\mathbb{I}_{AC}}{4}$. We see that $\tilde{\rho}_{AC\vert{}y}$ is positive semidefinite for all $y$, and the measurement operators are valid POVMs, so that $P(a,b,c\vert{}x,y,z) \ge 0$ for all $a,b,c,x,y,z$. In this case, since the state is fully mixed, we see that one has $\langle CHSH \rangle_{AC|y=1} = 0$.

\end{example}

This set is similar to the $\mathcal{RC}$ set considered in Appendix \ref{sec:RC} except that we now impose that the state is a positive semidefinite Hermitian operator $\rho_{ABC}$ of unit trace. The set differs from the standard set of quantum correlations $\mathcal{Q}$ in allowing for jamming maps $\Phi_y$ that purportedly still respect relativistic causality. 

To maintain non-negativity of probabilities, it is necessary to impose that $\Phi_y^{\dagger}(\rho_{AC}) \geq 0$. In Thm. \ref{thm:jammingmap-pos}, we have seen that if we instead require that $\Phi_y^{\dagger}$ is positive on every state (and not only the given initial state $\rho_{AC}$), then the condition of being locally identity preserving enforces that $\Phi_y = \mathbb{I}_{AC}$, i.e., no nontrivial positive jamming maps exist. We now consider the more restrictive condition that $\Phi_y^{\dagger}$ acts as a state-dependent positive map, i.e., only maintains positivity of the initial shared state $\rho_{AC}$. We prove that even this state-dependent version of jamming is unphysical, in that it reintroduces hidden superluminal signaling but this time via a different mechanism akin to that noted by Gisin \cite{Gisin90, SBG01} for nonlinear modifications of quantum theory.

We first recall the following Lemma on no-superluminal-signalling implying linearity of evolution, adapted to our particular evolution under jamming maps $\Phi_y^{\rho_{AC}}$. The proof follows arguments by Gisin in \cite{Gisin90} and Simon, Buzek and Gisin in \cite{SBG01} (see also the rederivation in \cite{BH15}). As usual, here an ensemble for a density operator \(\rho\in\mathcal{D}(\mathcal{H})\) is a finite collection \(\{(\rho^i,p_i)\}_{i\in I}\) where each \(\rho^i\in\mathcal{D}(\mathcal{H})\), \(p_i\ge0\), \(\sum_ip_i=1\), and
$\sum_ip_i\rho^i=\rho$. Furthermore, let \(|\Psi\rangle\in\mathcal{H}\otimes\mathcal{H}_D\) be a purification of \(\rho\) (i.e., \(\operatorname{tr}_D|\Psi\rangle\langle\Psi|=\rho\)). A projective measurement \(\{P_i\}_{i\in I}\) on \(\mathcal{H}_D\) steers the ensemble \(\{(\rho^i,p_i)\}\) if
\begin{equation}
   p_i=\langle\Psi|I\otimes P_i|\Psi\rangle
   \quad\text{and}\quad
   \rho^i=\frac1{p_i}\operatorname{tr}_R\bigl[(I\otimes P_i)|\Psi\rangle\langle\Psi|\bigr]
\end{equation}
for \(p_i>0\). 

\begin{assumption}
\label{ass:NS-assumption}
Let $\mathcal{H}_A$ and $\mathcal{H}_C$ be finite-dimensional Hilbert spaces and write $\mathcal{H}_{AC} = \mathcal{H}_A \otimes \mathcal{H}_C$. Let $\{\Phi_y^{\rho_{AC}}\}_{\rho_{AC} \in \mathcal{D}(\mathcal{H}_{AC}),y}$ be a family of linear, trace preserving maps. For every density operator $\rho_{AC} \in \mathcal{D}(\mathcal{H}_{AC})$ describing Alice-Charlie's system, every finite-dimensional reference system $\mathcal{H}_D$ that may be held by a space-like separated party David, and every purification $|\Psi_{ACD} \rangle \in \mathcal{H}_{AC} \otimes \mathcal{H}_D$ of $\rho_{AC}$, the following holds. 

Let $\{P_i\}_{i \in I}$ and $\{Q_j\}_{j \in J}$ be any two projective measurements on $\mathcal{H}_D$ that steer ensembles $\{(\rho^i_{AC}, p_i)\}$ and $\{(\sigma^j_{AC}, q_j)\}$ of $\rho_{AC}$ respectively. After the evolution, the states of the Alice-Charlie system must be identical:
\begin{eqnarray}
\label{eq:NS-ensembles}
    \sum_{i \in I} p_i (\Phi_y^{\rho^{i}_{AC}})^{\dagger}(\rho^{i}_{AC}) =  \sum_{j \in J} q_j (\Phi_y^{\sigma^{j}_{AC}})^{\dagger}(\sigma^{j}_{AC}).
\end{eqnarray}
Specifically, if the two mixtures in \eqref{eq:NS-ensembles} were different, a measurement performed on Alice-Charlie's system would reveal which of the two measurements $\{P_i\}_{i \in I}$ and $\{Q_j\}_{j \in J}$ were chosen by David on the spacelike separated system, thereby violating the principle of no-superluminal-signaling.   
\end{assumption}

\begin{lemma}
\label{lem:Gisin}
   Let $\mathcal{H}_A$ and $\mathcal{H}_C$ be finite-dimensional Hilbert spaces and write $\mathcal{H}_{AC} = \mathcal{H}_A \otimes \mathcal{H}_C$. Let $\{\Phi_y^{\rho_{AC}}\}_{\rho_{AC} \in \mathcal{D}(\mathcal{H}_{AC}),y}$ be a family of linear, trace preserving maps. Under Assumption \ref{ass:NS-assumption}, for every pair of density operators $\rho^0_{AC}, \rho^1_{AC} \in \mathcal{D}(\mathcal{H}_{AC})$, every $p \in [0,1]$ and every input $y$, it holds that
\begin{equation}
\label{eq:NS-equals-affinity}
    \left(\Phi_y^{p \rho^0_{AC} + (1-p) \rho^1_{AC}}\right)^{\dagger}\left(p \rho^0_{AC} + (1-p) \rho^1_{AC}\right) = p \; \left(\Phi_y^{\rho^0_{AC}}\right)^{\dagger}(\rho^0_{AC}) + (1-p) \; \left(\Phi_y^{\rho^1_{AC}}\right)^{\dagger}(\rho^1_{AC}).
\end{equation}
\end{lemma}
\begin{proof}
We are given arbitrary $\rho^0_{AC}, \rho^1_{AC} \in \mathcal{D}(\mathcal{H}_{AC})$ and $p \in [0,1]$. Let us define 
\begin{equation}
\label{eq:rhoAC-def}
\rho_{AC} := p \rho^0_{AC} + (1-p) \rho^1_{AC}. 
\end{equation}
By the GHJW theorem \cite{HJW93, Gisin89}, for any such $\rho^0_{AC}, \rho^1_{AC}$ and $p$, there exists a purification $|\Psi_{ACD} \rangle \in \mathcal{H}_{AC} \otimes \mathcal{H}_D$ of $\rho_{AC}$ and a two-outcome projective measurement $\{P_0, P_1\}$ on $\mathcal{H}_D$ that steers the ensemble $\{(\rho^0_{AC}, p), (\rho^1_{AC}, (1-p)\}$, while the single ensemble $\{\rho_{AC}, I_{AC}\}$ is steered by the trivial measurement $I_{D}$. By the no-superluminal-signaling condition in \eqref{eq:NS-ensembles} of Assumption \ref{ass:NS-assumption} applied to these two measurements $\{P_0, P_1\}$ and $\{I_D\}$, we have that the states after the evolution must be identical, i.e., it holds that
\begin{equation}
    \left(\Phi_y^{\rho_{AC}}\right)^{\dagger}(\rho_{AC}) = p \; \left(\Phi_y^{\rho^0_{AC}}\right)^{\dagger}(\rho^0_{AC}) + (1-p) \; \left(\Phi_y^{\rho^1_{AC}}\right)^{\dagger}(\rho^1_{AC}),
\end{equation}
which is equivalent to \eqref{eq:NS-equals-affinity} by \eqref{eq:rhoAC-def}.
\end{proof}

The following theorem shows that under the no-superluminal-signaling Assumption \ref{ass:NS-assumption} and the locally identity-preserving constraints on the maps $\Phi_y^{\rho_{AC}}$, the only maps that survive are those for which the effective state transformation is trivial: $\left(\Phi_y^{\rho_{AC}}\right)^{\dagger}(\rho_{AC}) = \rho_{AC}$ for all $\rho_{AC}, y$. In other words, no nontrivial jamming is possible even in this state-dependent setting. 

\begin{theorem}
    \label{thm:jamq-signal-app}
    Let $\mathcal{H}_A$ and $\mathcal{H}_C$ be finite-dimensional Hilbert spaces and write $\mathcal{H}_{AC} = \mathcal{H}_A \otimes \mathcal{H}_C$. Let $\{\Phi_y^{\rho_{AC}}\}_{\rho_{AC} \in \mathcal{D}(\mathcal{H}_{AC}),y}$ be a family of maps $\Phi_y^{\rho_{AC}} : \mathcal{L}(\mathcal{H}_{AC}) \to \mathcal{L}(\mathcal{H}_{AC})$ that satisfies, for every density operator $\rho_{AC}$ and every input $y$,
    \begin{itemize}
        \item $\Phi_y^{\rho_{AC}}$ is linear,
        \item $\Phi_y^{\rho_{AC}}$ is trace-preserving $\text{tr}\left(\Phi_y^{\rho_{AC}}(X)\right) = \text{tr}(X)$ for every $X \in \mathcal{L}(\mathcal{H}_{AC})$,
        \item $\Phi_y^{\rho_{AC}}$ is locally identity-preserving, i.e., $\Phi_y^{\rho_{AC}}(X_A \otimes \mathbb{I}_C) = X_A \otimes \mathbb{I}_C$ and $\Phi_y^{\rho_{AC}}(\mathbb{I}_A \otimes X_C) = \mathbb{I}_A \otimes X_C$ for all observables $X_A \in \mathcal{L}(\mathcal{H}_A)$ and $X_C \in \mathcal{L}(\mathcal{H}_C)$, 
        \item $\left(\Phi_y^{\rho_{AC}}\right)^{\dagger}(\rho_{AC}) \geq 0$.
    \end{itemize}
Then under the no-superluminal-signaling Assumption \ref{ass:NS-assumption}, it holds that
\begin{equation}
    \left(\Phi_y^{\rho_{AC}}\right)^{\dagger}(\rho_{AC}) = \rho_{AC},
\end{equation}
for every $\rho_{AC} \in \mathcal{D}(\mathcal{H}_{AC})$ and every input $y$.    
\end{theorem}
\begin{proof}
Fix an input $y$, since $y$ is arbitrary it is sufficient to prove the claim for this fixed value. Define $F_y : \mathcal{D}(\mathcal{H}_{AC}) \to \mathcal{L}(\mathcal{H}_{AC})$ by 
\begin{equation}
    F_y(\rho_{AC}) := \left(\Phi_y^{\rho_{AC}}\right)^{\dagger}(\rho_{AC}).
\end{equation}
By assumption, $F_y(\rho_{AC}) \geq 0$. From Lemma \ref{lem:Gisin}, for any ensemble $\rho_{AC} = \sum_i p_i \rho^i_{AC}$, the no-superluminal-signaling Assumption \ref{ass:NS-assumption} implies that
\begin{eqnarray}
    F_y(\rho_{AC}) = \sum_i p_i F_y(\rho^i_{AC}),
\end{eqnarray}
i.e., $F_y$ is affine on the convex set $\mathcal{D}(\mathcal{H}_{AC})$. 

We now show that $F_y$ preserves both reduced states on $\mathcal{H}_A$ and $\mathcal{H}_C$. To that end, observe that for every $X_A \in \mathcal{L}(\mathcal{H}_A)$, we have
\begin{equation}
\begin{aligned}
    \text{tr}\left[F_y(\rho_{AC})(X_A \otimes \mathbb{I}_C) \right] 
    &=\text{tr}\left[\left(\Phi_y^{\rho_{AC}}\right)^{\dagger}(\rho_{AC})(X_A \otimes \mathbb{I}_C) \right]
    &=\text{tr}\left[\rho_{AC} \Phi_y^{\rho_{AC}}(X_A \otimes \mathbb{I}_C) \right]
    &= \text{tr}\left[\rho_{AC}(X_A \otimes \mathbb{I}_C) \right].
\end{aligned}
\end{equation}
In the above, the first equality is the definition of $F_y$, the second equality is by the definition of the adjoint, and the third equality is by the local identity preservation of $\Phi_y^{\rho_{AC}}$. We therefore obtain that  
\begin{eqnarray}
\label{eq:local-identity-A}
 &\text{tr}\left[F_y(\rho_{AC})(X_A \otimes \mathbb{I}_C) \right] = \text{tr}\left[\rho_{AC}(X_A \otimes \mathbb{I}_C) \right] \nonumber \\
& \implies   \text{tr}\left[\text{tr}_C F_y(\rho_{AC}) X_A - \text{tr}_C \rho_{AC}X_A \right] = 0 \; \quad \forall X_A \in \mathcal{L}(\mathcal{H}_A) \nonumber \\
 &\implies \text{tr}_C F_y(\rho_{AC}) = \text{tr}_C \rho_{AC}.
\end{eqnarray}
By an analogous argument using $\Phi_y^{\rho_{AC}}(\mathbb{I}_A \otimes X_C) = \mathbb{I}_A \otimes X_C$ we obtain that 
\begin{equation}
\label{eq:local-identity-C}
    \text{tr}_A F_y(\rho_{AC}) = \text{tr}_A \rho_{AC}.
\end{equation}

We now proceed to show that $F_y$ fixes every pure product state. To that end, let $\sigma_{AC} = |\psi \rangle \langle \psi|_A \otimes |\phi \rangle \langle \phi|_C$ be an arbitrary pure product state, and set $\tau_{AC} := F_y(\sigma_{AC})$. By assumption, we have that $\tau \geq 0$. From the properties \eqref{eq:local-identity-A} and \eqref{eq:local-identity-C} derived above we have that
\begin{eqnarray}
    \text{tr}_{C} \tau_{AC} &=& |\psi \rangle \langle \psi|_A, \nonumber \\
    \text{tr}_{A} \tau_{AC} &=& |\phi \rangle \langle \phi|_C.
\end{eqnarray}
Since $\tau_{AC} \geq 0$ and $\text{tr}_{C} \tau_{AC} = |\psi \rangle \langle \psi|_A$ is pure, we obtain that $\tau_{AC} = |\psi \rangle \langle \psi|_A \otimes \sigma_C$ for some state $\sigma_C$. By the second marginal condition, we then obtain that $\sigma_C = \text{tr}_A \tau_{AC} = |\phi \rangle \langle \phi|_C$ giving $\tau_{AC} = |\psi \rangle \langle \psi|_A \otimes |\phi \rangle \langle \phi|_C$. We therefore obtain that
\begin{equation}
\label{eq:pure-state-fixing}
    F_y(\sigma_{AC}) = | \psi \rangle \langle \psi|_A \otimes |\phi \rangle \langle \phi|_C = \sigma_{AC}.
\end{equation}
In other words, $F_y$ fixes every pure product state. It remains to show that this implies that $F_y$ is the identity on all density operators.


Now since $\mathcal{H}_A$ and $\mathcal{H}_C$ are finite-dimensional, the real vector spaces $\text{Herm}(\mathcal{H}_A)$ and $\text{Herm}(\mathcal{H}_C)$ are also finite-dimensional and the pure-state projectors span these spaces (every Hermitian operator on $\mathcal{H}_A$ is a real linear combination of projectors $|\psi \rangle \langle \psi|_A$, and similarly for $\mathcal{H}_C$). It follows that the operators $|\psi \rangle \langle \psi|_A \otimes |\phi \rangle \langle \phi|_C$ span the real vector space $\text{Herm}(\mathcal{H}_A \otimes \mathcal{H}_C)$. Now, since $\mathcal{D}(\mathcal{H}_{AC})$ has nonempty relative interior in the real affine space of trace-one Hermitian operators, $F_y$ admits a unique real-linear extension $\tilde{F}_y$ to $\text{Herm}(\mathcal{H}_{AC})$. Furthermore, since $\tilde{F}_y$ fixes every pure product projector, by linearity we have that $\tilde{F}_y = \mathbb{I}_{AC}$ on $\text{Herm}(\mathcal{H}_{AC})$. Therefore, we obtain that 
\begin{equation}
    F_y(\rho_{AC}) = \left( \Phi_y^{\rho_{AC}}\right)^{\dagger}(\rho_{AC}) = \rho_{AC},
\end{equation}
for every density operator $\rho_{AC}$. 
\end{proof}
Therefore, we conclude that no nontrivial jamming extension of quantum correlations is possible - any attempt to introduce nontrivial jamming maps in either state-independent or state-dependent fashion leads to hidden superluminal signaling.

\section{Jamming Attacks on Device-Independent Cryptographic Protocols}
In this section, we investigate the security of Device-Independent (DI) protocols against adversaries constrained only by the principle of Relativistic Causality. As we have seen, the set $\mathcal{RC}$ is the correct set to consider when considering relativistic attacks on DI protocols rather than the hitherto considered no-signaling polytope $\mathcal{NS}$ \cite{MAG06, PR94, PBS11}.

\subsection{Jamming attacks on Device-Indpendent Bit Commitment}
\label{sec:Jam-BitCommitment}
In this section, we outline how jamming correlations affect the security of device-independent schemes for bit commitment. Specifically, we show how the two-prover bit commitment scheme of Fehr-Fillinger \cite{FF15} which was proven secure solely on the impossibility of superluminal signaling (when the devices held by the dishonest provers correspond to an arbitrary no-signaling behavior), fails to be secure in configurations which allow jamming. 

Recall that a bit commitment scheme lets a prover (committer) lock in a bit $b \in \{0,1\}$ such that two properties are obeyed:
\begin{itemize}
    \item Hiding: after the commit phase, a verifier has (almost) no information about $b$,
    \item Binding: the prover cannot later open the commitment to both $0$ and $1$.
\end{itemize}
It is well-known that in the classical and quantum single-prover scenarios it is impossible to achieve both properties unconditionally \cite{Mayers97, LoChau97}. Ben-Or, Goldwasser, Kilian and Wigderson \cite{BGKW88} therefore introduced the multi-prover model, whereby the prover is split into two or more agents that are assumed unable to communicate with each other once the protocol begins. The hope then is that the no-communication assumption alone is sufficient to enforce binding. 

Considering two-prover schemes, it is also known that every single-round two-prover commitment scheme that is almost perfectly hiding can be broken by a no-signaling attack. Further, the attack is essentially optimal: the dishonest provers can open the commitment to either bit with success probability equal to the honest opening probability (probability $1$ when the scheme is perfectly sound. 

On the other hand, as shown by Fehr-Fillinger, there exists a three-prover commitment scheme that is secure against every non-signaling attack \cite{FF15}. Here, we show that for the same scheme is rendered fully insecure by jamming attacks in the specific spacetime configurations where the third prover $P_3$ can jam the first two provers $P_1$ and $P_2$ (i.e., where the intersection of the future light cones of the output events of $P_1$ and $P_2$ lies within the future light cone of the input event of $P_3$).  

We briefly sketch the (simple) argument before presenting a formal proof in Thm. \ref{thm:RCattack-BitCommitment}. For a scheme to be considered binding against every no-signaling (NS) strategy, the combined probability of successfully opening to bit $0$ and bit $1$ must be strictly limited as $p_0^{NS} + p_1^{NS} \leq 1+ \eta$ for some $\eta < 1$ \cite{FF15}. Now, RC constraints in the jamming configuration permit the joint distribution of messages from $P_1, P_2$ to depend on an auxiliary classical bit $c$ given to $P_3$. This allows a strategy whereby both $P_1$ and $P_2$ sample and store opening strings for both possible bits using shared randomness, while $P_3$ uses $c$ to sample and output the corresponding marginals for the chosen string. By executing this RC strategy, the provers can successfully open the commitment to either bit with a combined probability of $p_0^{RC} + p_1^{RC} \geq 2 - 2 \epsilon - 2 \delta$ where $\epsilon$ is the hiding parameter and $\delta$ is the soundness error. In an ideal commitment scheme with $\epsilon, \delta \to 0$, this success rate exceeds the binding bound $1 + \eta$ that holds against NS adversaries, showing that RC attacks successfully break the commitment scheme. 

To establish this formally, let us define the three-prover bit commitment scheme $\Pi$ of \cite{FF15}. The verifier samples a challenge tuple $a = (a_1, a_2, a_3)$ according to a prior distribution $\mu(a)$ and sends $a_i$ to prover $P_i$. Each prover $P_i$ replies with a commit-phase message $x_i \in X_i$. To open a committed bit $b \in \{0, 1\}$, the provers subsequently return opening messages $y_i \in Y_i$. The verifier evaluates a predicate $V(a, x, y, b) \in \{0, 1\}$ and accepts if and only if $V = 1$. Let $p_b^{\mathrm{hon}}(x, y \vert{} a)$ denote the joint probability distribution of all messages when the provers execute the honest strategy to commit to bit $b$ and later open to $b$. The scheme satisfies three properties:

\begin{itemize}
\item $\varepsilon-Hiding$: The commit-phase messages are statistically indistinguishable to the verifier: $d(p_0^{\mathrm{hon}}(x), p_1^{\mathrm{hon}}(x)) \leq \varepsilon$, where $d(P, Q) = \frac{1}{2}\sum_x \vert{}P(x) - Q(x)\vert{}$ is the total variation distance.

\item $\delta-Sound$: The honest strategies succeed with high probability: $\mathbb{E}_{a \sim \mu} \bigl[ p_b^{\mathrm{hon}}(V(a, x, y, b) = 1) \bigr] \ge 1 - \delta \quad \forall b \in \{0, 1\}$.

\item Binding against Non-Signaling: Against any non-signaling (NS) adversary, the sum of success probabilities for opening to $0$ and $1$ is bounded as $p_0^{\mathrm{NS}} + p_1^{\mathrm{NS}} \le 1 + \eta$ for some $\eta < 1$. 
\end{itemize}

\begin{theorem}
\label{thm:RCattack-BitCommitment}
Suppose that the three provers $P_1, P_2$ and $P_3$ executing protocol $\Pi$ above are embedded in a spacetime configuration $\mathcal{M}$ where $P_3$ is permitted to jam the pair $\{P_1, P_2\}$. There exists a behavior in the Relativistically Causal (RC) set of $\mathcal{M}$ such that the provers can successfully open to both $b=0$ and $b=1$ with a combined probability:
\begin{equation}
p_0^{\mathrm{RC}} + p_1^{\mathrm{RC}} \ge 2 - 2\varepsilon - 2\delta,
\end{equation}
where $\epsilon$ is the hiding parameter and $\delta$ is the soundness error of the protocol $\Pi$. In the ideal scheme with $\epsilon, \delta \to 0$, this breaks the NS binding bound without violating relativistic causality.
\end{theorem}

\begin{proof}
We first recall the following lemma from \cite{FF15}. 
\begin{lemma}[\cite{FF15}]
\label{lem:gluing-lemma}
Let $p_0(u, v)$ and $p_1(u, v)$ be two probability distributions defined on a finite product space $\mathcal{U} \times \mathcal{V}$ such that the total variation distance between their marginals on $\mathcal{U}$ is bounded as $d(p_0(u), p_1(u)) \leq \varepsilon$. 
Then there exists a joint distribution $q(u_0, u_1, v_0, v_1)$ such that:
\begin{eqnarray}
\sum_{u_1, v_1} q(u_0, u_1, v_0, v_1) &=& p_0(u_0, v_0) \nonumber \\
\sum_{u_0, v_0} q(u_0, u_1, v_0, v_1) &=& p_1(u_1, v_1) \nonumber \\
q(u_0 \neq u_1) &\leq& \varepsilon. 
\end{eqnarray}
\end{lemma}
\begin{proof} 
By definition of statistical distance, there is a coupling $r(u_0, u_1)$ of the marginals $p_0(u)$ and $p_1(u)$ (with $\sum_{u_1} r(u_0, u_1) = p_0(u_0)$ and $\sum_{u_0} r(u_0, u_1) = p_1(u_1)$) satisfying $P_r(u_0 \neq u_1) = \sum_{u_0, u_1: u_0 \neq u_1} r(u_0, u_1) = \varepsilon$. We set:
\begin{equation}
q(u_0, u_1, v_0, v_1) := r(u_0, u_1) \, p_0(v_0 \vert{} u_0) \, p_1(v_1 \vert{} u_1).
\end{equation}
The three properties hold by direct inspection of the marginals. \end{proof}

Now, for each fixed challenge $a$, apply the Lemma \ref{lem:gluing-lemma} to the two honest distributions $p_0^{\mathrm{hon}}(\cdot \vert{} a)$ and $p_1^{\mathrm{hon}}(\cdot \vert{} a)$, taking $u = x$ and $v = y$. This produces a family of distributions $q(x^{(0)}, x^{(1)}, y^{(0)}, y^{(1)} \vert{} a)$ such that, for every $a$ \cite{FF15}:
\begin{itemize}
\item The marginal on $(x^{(0)}, y^{(0)})$ equals $p_0^{\mathrm{hon}}(\cdot \vert{} a)$,

\item The marginal on $(x^{(1)}, y^{(1)})$ equals $p_1^{\mathrm{hon}}(\cdot \vert{} a)$,

\item $q(x^{(0)} \neq x^{(1)} \vert{} a)$ satisfies $\mathbb{E}_{a \sim \mu}[q(x^{(0)} \neq x^{(1)})] \leq \varepsilon$.
\end{itemize}

Define a new distribution $r(x, y^{(0)}, y^{(1)} \vert{} a)$ on the set of triples $(x, y^{(0)}, y^{(1)})$ by:
\begin{equation}
r(x, y^{(0)}, y^{(1)} \vert{} a) := q(x^{(0)}=x, x^{(1)}=x, y^{(0)}, y^{(1)} \vert{} a) + \frac{1}{2} q(x^{(0)} \neq x^{(1)}, y^{(0)}, y^{(1)} \vert{} a).
\end{equation}
We now compute the acceptance probabilities under $r$. On the event $x^{(0)} = x^{(1)} = x$, the pair $(x, y^{(0)})$ is distributed exactly as $p_0^{\mathrm{hon}}$ and the pair $(x, y^{(1)})$ is distributed exactly as $p_1^{\mathrm{hon}}$. On the complementary event (which has probability at most $\varepsilon$), the acceptance probability is non-negative. Therefore, we obtain
\begin{eqnarray}
\mathbb{E}_{a \sim \mu} \bigl[ r(V(a, x, y^{(0)}, 0) = 1) \bigr] &\geq& (1-\varepsilon)(1-\delta) + \varepsilon \cdot 0 \nonumber \\
&=& 1 - \delta - \varepsilon + \varepsilon\delta \geq 1 - \delta - \varepsilon.
\end{eqnarray}
Similarly, for $b=1$:
\begin{eqnarray}
\mathbb{E}_{a \sim \mu} \bigl[ r(V(a, x, y^{(1)}, 1) = 1) \bigr] \geq 1 - \delta - \varepsilon.
\end{eqnarray}
We now embed this into an RC behavior via the following construction:
\begin{itemize}
\item Using shared randomness that realizes the distribution $r(\cdot \vert{} a)$, the provers $P_1$ and $P_2$ sample the tuple $(x_1, x_2, y_1^{(0)}, y_1^{(1)}, y_2^{(0)}, y_2^{(1)})$ and output the commit-phase messages $x_1, x_2$.

\item Prover $P_3$ receives their challenge $a_3$ and the auxiliary bit $c$. $P_3$ outputs $x_3$ together with the single opening string $y_3^{(c)}$.

\item In the opening phase, $P_1$ and $P_2$ output the strings $y_1^{(c)}$ and $y_2^{(c)}$.
\end{itemize}
The construction ensures that the resulting joint distribution of all messages, conditioned on $c$, is exactly the corresponding marginal of $r$. Furthermore, the single-party marginals are independent of $c$ by construction. Therefore, the behavior lies within the set $\mathcal{RC}$ of the jamming configuration.

By construction, it holds that if the auxiliary bit $c = 0$, the verifier's predicate for bit $0$ evaluates to $1$ with probability at least $1 - \delta - \varepsilon$. Analogously, if the auxiliary bit  $c = 1$, the predicate for bit $1$ evaluates to $1$ with probability at least $1 - \delta - \varepsilon$. Therefore we obtain,
\begin{equation}
p_0^{\mathrm{RC}} + p_1^{\mathrm{RC}} \ge 2 - 2\varepsilon - 2\delta.
\end{equation}
We thus obtain a behavior that surpasses the NS binding bound without violating relativistic causality.
\end{proof}

\subsection{Jamming attacks on Device-Independent Secret Sharing}
Device-Independent (DI) Secret Sharing aims to distribute a secret among multiple receivers so that the secret can be reconstructed only when a designated subset of receivers collaborate, while remaining hidden from any proper subset and from external eavesdroppers. The protocol of \cite{MBNC20} involves three distant parties - Charlie (the sender) and Alice and Bob (the receivers), who all share a device that takes binary inputs $x, y, z$ and produces binary outputs $a, b, c$ respectively. For a fixed input triple $(x^*, y^*, z^*) = (0,0,0)$, the outputs are required to satisfy the relation $c = a \oplus b$ (where $a \oplus b$ denotes the XOR of the two bits $a,b$). So that the sender Charlie's secret bit $c$ is the parity of the two receivers' bits $a$ and $b$ which appear random and hold no information individually. The protocol is required to be secure against an external eavesdropper but also against an untrusted receiver Alice, who may produce an additional output $e$ in an attempt to guess the secret, i.e., the dishonest Alice aims to maximize her guessing probability of the secret $P_{\text{guess}}(e=c|x^*,y^*,z^*)$ without collaborating with Bob. Security against such attacks is certified by the parties testing for sufficiently strong genuine multipartite nonlocality in the form of the violation of Svetlichny inequality $S \leq 4$ \cite{Svetlichny87} by the observed behavior $\{P(a,b,c|x,y,z)\}$. Formally, the Svetlichny inequality is given as
\begin{eqnarray}
\label{eq:Svetlichny}
S &=& \sum_{a,b,c = 0}^1 (-1)^{a+b+c} \Big[ 
P(a,b,c|0,0,0) + P(a,b,c|0,0,1) + P(a,b,c|0,1,0) - P(a,b,c|0,1,1) \nonumber \\
&& \qquad \qquad + P(a,b,c|1,0,0) - P(a,b,c|1,0,1) - P(a,b,c|1,1,0) - P(a,b,c|1,1,1) 
\Big] \leq 4.
\end{eqnarray}
Under the full set of no-signaling constraints (i.e., that the observed behavior lies within $\mathcal{NS}$) the maximal violation of the Svetlichny inequality enforces that $P_{\text{guess}}(e=c|x^*,y^*,z^*) = 1/2$.

\begin{theorem}
    Let $\Pi$ denote the tripartite DI secret sharing protocol of \cite{MBNC20} explained above. Suppose that the protocol is executed in a jamming configuration where Bob's measurement event lies in the causal past of the intersection of the future light cones of Alice and Charlie, so that Bob is permitted to jam the Alice-Charlie marginal. Then there exists a behavior $\{P(a,b,c, e|x,y,z)\}$ in the relativistically causal set $\mathcal{RC}$ such that
    \begin{itemize}
        \item the observable behavior with marginals $P(a,b,c|x,y,z) = \sum_e P(a,b,c,e|x,y,z)$ attains the algebraic maximum of the Svetlichny expression \eqref{eq:Svetlichny} $S = 8$,
        \item a dishonest Alice is able to perfectly guess the secret with an auxiliary output $e$, i.e., Alice achieves $P_{\text{guess}}(e=c|x^*,y^*,z^*) = 1$ without collaborating with Bob. 
    \end{itemize}
\end{theorem}

\begin{proof}
The explicit behavior $\{P(a,b,c,e|x,y,z)\} \in \mathcal{RC}$ achieving the two conditions stated in the theorem is computed by optimizing the attack strategies over the convex polytope $\mathcal{RC}$ and is given by
\begin{equation}
    P(a,b,c,e|x,y,z) = \frac{1}{4} \; \delta_{a \oplus b \oplus c, xy \oplus xz \oplus yz} \cdot \delta_{e,a \oplus b},
\end{equation}
where $\delta_{i,j}$ denotes the Kronecker delta  ($\delta_{i,j} = 1 1$ if \(i = j\) and $0$ otherwise), and \(\oplus \) denotes addition modulo $2$.

Firstly, we sum over the auxiliary output $e$ to obtain the observed behavior as  
\begin{equation}
    P(a,b,c|x,y,z) = \sum_e P(a,b,c,e|x,y,z) = \frac{1}{4} \; \delta_{a \oplus b \oplus c, xy \oplus xz \oplus yz}. 
\end{equation}
We see that for $(x,y,z) \in \{(0,0,0), (0,0,1), (0,1,0), (1,0,0)\}$ we have $xy \oplus xz \oplus yz = 0$ giving $P(a,b,c|x,y,z) = \frac{1}{4}$ for $a \oplus b \oplus c = 0$ meaning $(-1)^{a+b+c} = 1$. Similarly for $(x,y,z) \in \{(0,1,1),(1,0,1), (1,1,0), (1,1,1)\}$ we have $xy \oplus xz \oplus yz = 1$ giving $P(a,b,c|x,y,z) = \frac{1}{4}$ for $a \oplus b \oplus c = 1$ meaning $(-1)^{a+b+c} = -1$. Substituting in \eqref{eq:Svetlichny} gives $S = 8$. 

Secondly, we verify that the behavior $\{P(a,b,c,e|x,y,z)\}$ satisfies all the relativistic causality constraints for the jamming configuration. Recall that \cite{HR19} here all marginal behaviors of any subset of parties is required to be independent of the inputs of the complementary set, except the joint behavior of Alice-Charlie which may depend upon the input of Bob. We verify the relativistic causality constraints by explicitly computing the marginal probabilities as
\begin{eqnarray}
    P(a,e|x,y,z) &=& \sum_{b,c} P(a,b,c,e|x,y,z) = \frac{1}{4} \; \quad \text{independent of} \; y, z, \nonumber \\
    P(b|x,y,z) &=& \sum_{a,c,e} P(a,b,c,e|x,y,z) = \frac{1}{2} \; \quad \text{independent of} \; x,z, \nonumber \\
    P(c|x,y,z) &=& \sum_{a,b,e} P(a,b,c,e|x,y,z) = \frac{1}{2} \; \quad \text{independent of} \; x,y, \nonumber \\
    P(a,e,b|x,y,z) &=& \sum_{c} P(a,b,c,e|x,y,z) = \frac{1}{4} \; \delta_{e, a \oplus b} \; \quad \text{independent of} \; z, \nonumber \\
    P(b,c|x,y,z) &=& \sum_{a,e} P(a,b,c,e|x,y,z) = \frac{1}{4} \; \quad \text{independent of} \; x.
\end{eqnarray}
On the other hand, we also compute that
\begin{eqnarray}
    P(a,c,e|x,y,z) &=& \sum_{b} P(a,b,c,e|x,y,z) = \frac{1}{4} \, \delta_{e \oplus c, \, xy \oplus xz \oplus yz} \quad \text{dependent on} \; y,
\end{eqnarray}
showing that the behavior $\{P(a,b,c,e|x,y,z)\} \in \mathcal{RC}$ but $\{P(a,b,c,e|x,y,z)\} \notin \mathcal{NS}$.

Finally, we have by construction that $e = a \oplus b$ in the behavior. For the input $(x^*,y^*,z^*) = (0,0,0)$ one has $x^*y^* \oplus x^* z^* \oplus y^* z^* = 0$ giving $a \oplus b \oplus c = 0$. Under this condition $e = a \oplus b$ ensures $e=c$ for this input, and we conclude that $P_{\text{guess}}(e=c|x^*,y^*,z^*) = 1$.

\end{proof}

\end{document}